%% file: main.tex
\documentclass[11pt]{article}
\input{preamble}

\makeatletter
\@for\reviewenv:=theorem,lemma,corollary,proposition,definition,remark,claim,invariant,observation,example,fact,conjecture,problem\do{%
  \expandafter\def\csname theH\reviewenv\endcsname{\thesection.\arabic{theorem}}}
\makeatother

\begin{document}

\allowdisplaybreaks

\title{Complexity of Differentially Private Selection via Federated APIs}
\author{\hspace{0.3in}
Hilal Asi\thanks{Apple Inc. \texttt{hilal.asi94@gmail.com}.} \and Vitaly Feldman\thanks{Apple Inc. \texttt{vitaly.edu@gmail.com}.} \and Jelani Nelson\thanks{Anthropic. \texttt{minilek@alum.mit.edu}. Author contributed to the project while at UC Berkeley, supported by NSF grant CCF-1951384, ONR grant N00014-18-1-2562, and ONR DORECG award N00014-17-1-2127.} \and Huy L. Nguyen\thanks{Northeastern. \texttt{huylenguyen@gmail.com}.  Supported in part by NSF CCF-2311649.}
\hspace{0.3in}\and
Kunal Talwar\thanks{Apple Inc. \texttt{kunal@kunaltalwar.org}.} \and 
Samson Zhou\thanks{Texas A\&M University. \texttt{samsonzhou@gmail.com}. Author is supported in part by NSF CCF-2335411 and gratefully acknowledges funding provided by the Oak Ridge Associated Universities (ORAU) Ralph E. Powe Junior Faculty Enhancement Award.}}

\maketitle

\begin{abstract}
Differential privacy (DP) has become the standard for privacy-preserving data analysis, yet its practical implementation in federated or distributed settings often faces a significant utility gap compared to the central model. While existing frameworks, including local, shuffle, and secure aggregation models, offer vetted primitives, they can incur high error for important tasks such as selection. Our first contribution is to show that this is inherent in the aggregation model: any private selection algorithm relying on the standard noisy aggregation primitive requires $\Omega(\sqrt{d}/\varepsilon)$ samples. To address this limitation, we propose adding a new primitive that implements a version of the \textit{Sparse Vector Technique} (SVT) and argue that it has simple implementations under various trust models. We demonstrate its power by developing an $\varepsilon$-differentially private algorithm for the selection problem. For selection in dimension $d$, our algorithm achieves an additive error of $O\left(\frac{\log(dn)}{\varepsilon}\right)$, matching that of the central model when $n$ is polynomially bounded in $d$. Furthermore, we show that this can be done using a near-optimal expected number $O(d\log^*d)$ of SVT queries. Together with existing lower bounds for the shuffle model, this result establishes an exponential separation between the SVT model and shuffle or aggregation-based models. We demonstrate the model's broad applicability to diverse tasks such as $F_p$ loss minimization, heavy-hitter identification, and sparse mean estimation.
\end{abstract}

\section{Introduction}
\label{sec:intro}
Differential privacy (DP)~\cite{DworkMNS06j, DworkRBook} has emerged as the gold standard for privacy-preserving data analysis, offering a mathematically rigorous framework that limits the risk of disclosing the data of any individual entity. 
Informally, a randomized algorithm is differentially private if its output distribution is nearly unchanged when any individual record is modified or removed. 
This ensures that adversaries learn little about any one participant, even with auxiliary information.

Despite its mathematical elegance, the practical implementation and verification of differential privacy can be challenging~\cite{BartheO13,BartheDGKB13,GaboardiHHNP13,BartheGAHKS14,BartheFGGHS16,BartheGGHS16,ZhangK17}.
Because privacy guarantees are probabilistic and often counterintuitive, even subtle coding errors or floating-point vulnerabilities can lead to privacy leakage~\cite{Mironov12}. 
This has led the community toward primitives that avoid floating-point numbers~\cite{Canonne_Kamath_Steinke_2022, Balcer_Vadhan_2019}. 
More broadly, verifying that an implementation of a differentially private algorithm is correct can be challenging, and even testing or auditing for privacy requires considerable effort~\cite{BichselGDTV18,GuanFHY23}.

To mitigate these risks, the field has moved toward a modular ``API-centric'' approach. 
Rather than requiring developers to implement complex mechanisms from scratch, privacy frameworks, dating from early systems like PINQ \cite{McSherry10} and Airavat \cite{RoySKSW10} to modern efforts like OpenDP and Google's differential privacy building blocks, rely on a small set of vetted primitives. 
These primitives are the ``building blocks'' for private algorithms; by thoroughly auditing a limited number of functions, such as the Laplace mechanism or standard composition theorems, one can build complex pipelines while mathematically tracking the total privacy loss. 

In federated or distributed architectures, implementing such primitives is harder, and thus this modularity is even more critical.
Specifically, primitives must be implemented under a variety of operational and security constraints. 
The move toward these standardized, vetted interfaces has led to the adoption of large-scale DP frameworks~\cite{opendp,WilsonZLDSG20}, which provide high-level query languages backed by reliable primitives.
However, as the ecosystem of these frameworks expands, it becomes vital to understand the relative power of various primitives.

Previous work has successfully formalized specific trust models, such as the local~\cite{warner1965randomized,EvfimievskiGS03,KasiviswanathanLNRS11} and shuffle~\cite{BittauEMMRLRKTS17,CheuSUZZ19,ErlingssonFMRTT19} privacy models.
More broadly, a privacy framework can treat simple privacy primitives as interchangeable building blocks and support a set of privacy-accounting methods. A complex algorithm can be built by combining primitives, each run on all or a subset of the data, while privacy accounting upper bounds the algorithm's privacy cost.
As an example, an aggregation system allows for algorithms that apply a function $f : \calX \rightarrow \mathbb{R}^d$ to each individual data point and output a noisy aggregate $\tilde{f}(x_1,\ldots,x_n) = \sum_i \text{clip}_1(f(x_i)) + \mathcal{N}(0, \sigma^2 \mathbf{I}_d)$, 
for a user-specified noise parameter $\sigma$, where $\text{clip}_1$ is the standard function that truncates the input to have $\ell_2$ norm at most $1$. 
A R\'enyi accountant in this setting would accumulate privacy cost in terms of R\'enyi DP~\cite{Mironov17}. One can then use such a system to implement differentially private gradient descent to train a machine learning model. Enhancing the primitive and accounting to allow for privacy amplification by sampling further enables DP stochastic gradient descent~\cite{BassilyST14,AbadiCGMMTZ16}. Aggregation, potentially with sampling, can support a rich class of algorithms, and such systems have been proposed under various assumptions~\cite{DworkKMMN06, Bonawitz17, crypteps, unlynx, honeycrisp, orchard, Talwaretal24}.

In this work, we study the limits of the Aggregation API, and propose an additional primitive to help circumvent some of these limitations.

\section{Our Contributions}
Despite the wide applicability of the Aggregation API, there are important tasks that are not known to be implementable accurately using aggregation alone. Foremost among these is perhaps the selection problem, where the best-known algorithms using aggregation incur significant overhead compared to the optimal central-DP bounds: while centrally one can do selection using $O(\log d/\eps)$ samples, existing implementations using aggregation require $\Omega(\sqrt{d}/\eps)$ samples. While this sample complexity is provably required (sometimes under additional assumptions) for weaker APIs such as local and shuffle DP~\cite{DuchiJW13,CheuSUZZ19,CheuU21,ChenGKM21}, no existing lower bounds exclude a better way to do selection using an aggregation API.

Our first result is a strong lower bound ruling out this possibility. 
We show that {\em any} selection algorithm using an aggregation API requires sample complexity $\Omega(\sqrt{d}/\eps)$ if it guarantees $(\eps,\delta)$-DP for any implementation of the aggregation API and can identify a dimension whose average score is $1/3$ more than all the others.
\begin{restatable}{theorem}{thmlbsq}[Informal]
\label{thm:lb:sq}
Any $(O(1), O(1/\poly(d)))$-DP algorithm for selection from $d$ candidates using aggregation queries alone requires $\Omega(\sqrt{d})$ samples.
\end{restatable}
The main advantage of this lower bound over most existing lower bounds for selection in restricted models is that we allow algorithms that use queries adaptively. 
Adaptivity is crucial for implementing most standard ML algorithms and is known to be necessary for some learning problems~\cite{dagan2020interaction}. 
We prove this result specifically for aggregation with Gaussian noise addition, with the privacy guarantee derived by composition of the guarantees for individual queries (\Cref{thm:main-lower-bound}). This implies that our lower bound applies to any API with a comparable privacy-utility trade-off, which is known to be asymptotically optimal for $L_2$ error. We observe that such an aggregation oracle can be implemented via distributed noise addition~\cite{DworkKMMN06} while satisfying weaker privacy guarantees. We then prove a lower bound for such protocols using the information-theoretic techniques from~\cite{braverman2016communication,duchirogers2019lower}.

To address these limitations, we propose adding a new {\em Sparse Vector Technique} (SVT) primitive. This primitive implements a special case of the sparse vector technique of \cite{DworkNRRV09}, and suffices to implement a large class of private algorithms. 
By integrating SVT as a fundamental primitive, our approach expands the expressive power of DP APIs, improves accuracy guarantees, and supports a wider range of applications.

Informally, the SVT primitive implements the {\sc AboveThreshold} subroutine for additive queries. 
We first formally define the SVT primitive. 
This technique is designed to privately process a potentially long sequence of queries by revealing information only about those whose true values exceed a specified threshold. 
A key innovation of SVT is that by careful correlation of the noise, the incurred privacy cost grows with the number of initialized SVT streams, and is independent of the number of queries issued. 
In our proposed primitive, each query is restricted to be a sum of bounded contributions from individual users. 
These are often referred to as additive or counting queries in the context of SVT, and decomposable score functions in the context of selection. 
Given $n$ users with data points $x_1,\ldots,x_n \in \calX$, a query is specified by a function $q_j : \calX \rightarrow [0,1]$ and a threshold $\tau_j$. 
Up to noise, the query asks if $\sum_{i=1}^n q_j(x_i) \geq \tau_j$, and a {\sc Yes} answer closes the current stream. \footnote{Symmetrically, one may consider a below-threshold version of such a query. More generally, one can consider a {\sc BetweenThresholds} query~\cite{BunSU17} with two thresholds, where the stream closes when the query answer is between the two thresholds.}
Because SVT needs to be stateful, we wrap this functionality in an {\em object} that can be created once and queried multiple times with different $(q_j, \tau_j)$ values chosen potentially in an online fashion. 
As in SVT, the API ensures that {\sc No} answers do not incur an additional per-query cost, and each query is correctly answered up to a small additive error.

This primitive has the advantage of simplicity: each query is similar to an aggregation query with one additional threshold parameter. 
This simplicity allows for efficient and modular implementations of this primitive. While the precise implementation will vary with the trust model, we discuss several options in \Cref{subsec:implementation}.

At the same time, we show that this primitive is quite powerful. 
SVT by itself has several applications and one can naturally implement any algorithm that builds on SVT for additive queries in this model (e.g., Private Multiplicative Weights~\cite{HardtR10}). 
We next argue that selection (the task solved by the exponential mechanism~\cite{McSherryT07}) can be implemented using the SVT API for any decomposable score function. 
Recall that given a dataset $(x_1,\ldots,x_n) \in \calX^n$, and $d$ score functions $s_1,\ldots, s_d: \calX^n \rightarrow \R$, the goal is to select the candidate $k \in [d]$ that (approximately) maximizes the score. 
We restrict ourselves to additive score functions of the form $s_j(x_1,\ldots,x_n) = \sum_{i=1}^n f_j(x_i)$ for some $f_j : \calX \rightarrow [0,1]$. We say an output candidate $j$ achieves additive error $\alpha$ if
\[s_j(x_1,\ldots,x_n) \ge \max_{k\in[d]} s_k(x_1,\ldots,x_n) - \alpha.\]

The reader may immediately notice that if one knew the target value $\max_{k\in[d]} s_k(x_1,\ldots,x_n)$, then a sequential pass over the $d$ score functions with that target value minus the SVT accuracy margin as a threshold yields a candidate with additive error at most twice that margin. 
This gives a mechanism with error $O(\log d / \eps)$ that matches that of the exponential mechanism, and is known to be optimal. Without knowing the target value, the natural approach is to start with a large guess for this threshold, and gradually reduce it. 
This algorithm (\Cref{alg:gen:api:linear}) yields the optimal error with $O(d(n+1))$ SVT queries.

It is easy to see that any algorithm must make $\Omega(d)$ queries. When $n$ is large, the algorithm above incurs significant overhead in the number of queries. This leads to the natural question of whether one can reduce the number of queries to be closer to linear in $d$. 
We show that this is indeed possible: we give a new algorithm that achieves the optimal $O(\log d/\eps)$ additive error (when $n$ is polynomial in $d$) and makes only $O(d\log^\ast d)$ queries to the SVT API in expectation.
\begin{restatable}{theorem}{thmselectthresh}
\label{thm:select:thresh}
Let $\eps\in(0,1]$. 
There exists an $\eps$-pure DP protocol that solves the selection problem for additive scores up to expected additive error $O\left(\frac{\log(dn)}{\eps}\right)$, using $O(d\log^\ast n+\log^3 n)$ oracle queries to the SVT API in expectation.  
\end{restatable}

When the number of users $n$ is upper bounded by a polynomial in $d$, the error of our protocol matches the optimal additive error of $\O{\frac{\log d}{\eps}}$ achievable in the central model of privacy, and the query complexity matches the $d$ lower bound up to an $O(\log^\ast d)$ factor.
We then show a number of applications of \Cref{thm:select:thresh}, specifically, to $F_p$ losses, heavy hitters with categories, multiplicative weights with the exponential mechanism, and sparse mean estimation. 

Our lower and upper bounds for selection extend to other natural tasks. Combined with known results for the shuffle model, our new upper bounds in the SVT model, and our lower bounds in the aggregation model emphasize the value of the SVT primitive. 
As we show in \Cref{fig:time:summary}, our new API allows for accurate selection, where aggregation and shuffle models must incur large errors. 

\definecolor{darkblue}{rgb}{0,0,.75}
\newcommand{\darkblue}[1]{\textcolor{darkblue}{#1}}

\begin{figure*}[!htb]
\centering
\small
\setlength{\tabcolsep}{6pt}
\renewcommand{\arraystretch}{1.15}
\begin{tabular}{lcccc}
\toprule
 & \textbf{\darkblue{SecAgg}} 
 & \textbf{\darkblue{Shuffle DP}} 
 & \textbf{\darkblue{SVT}} 
 & \textbf{\darkblue{Central DP}} \\
\midrule
\textbf{Aggregation} 
& $\frac{\sqrt{d}}{\eps}$ 
& $\frac{\sqrt{d}}{\eps}$ 
& $\frac{\sqrt{d}}{\eps}$ 
& $\frac{\sqrt{d}}{\eps}$ \\

\textbf{Selection} 
& $\frac{\sqrt{d}}{\eps}$ 
& $\frac{\sqrt{d}}{\eps}$ 
& $\frac{\log d}{\eps}$ 
& $\frac{\log d}{\eps}$ \\


\bottomrule
\end{tabular}

\vspace{-0.08in}
\caption{Comparison of the additive error achievable for the different APIs in the federated setting, up to constant factors, for $n$ users. In the aggregation row each user holds a vector of $\ell_2$ norm at most $1$ and the error is the $\ell_2$ error of the sum. In the selection row each user holds scores in $[0,1]^d$ and the error is the additive error in the sum score; the entry $\sqrt{d}/\eps$ indicates that error $o(n)$ is not achievable unless $n\ge\Omega(\sqrt{d}/\eps)$. The entries for SVT and central DP are upper bounds ($\eps$-DP, \Cref{thm:select:thresh} for SVT with $n$ polynomial in $d$). The entries for SecAgg and shuffle DP are lower bounds: for algorithms that access the data through an aggregation API (\Cref{thm:main-lower-bound}), and for robust or single-round shuffle protocols (\Cref{lb:shuffle:one} and \Cref{lb:shuffle:two}).}
\vspace{-0.12in}
\label{fig:time:summary}
\end{figure*}


\textbf{Other related work.} 
Our lower bounds for aggregation APIs are most closely related to statistical query (SQ) lower bounds in learning theory \cite{BlumFJKMR94,fgrvx2017statistical}. An SQ oracle is an oracle that for a query $q$ and data distribution $P$ is guaranteed to return a value within some prescribed tolerance $\tau$ of the expectation $\E_{x\sim P}[q(x)]$ (the returned value may be chosen arbitrarily by the oracle). It is known that learning a parity function over $k$ coordinates requires $\Omega(2^k \tau^2)$ statistical queries of tolerance $\tau$ (e.g.~ \cite{Feldman14}). Learning a parity function reduces to selection in dimension $2^k$ and thus this also gives a lower bound for selection. These lower bounds do not apply to our setting as we answer queries relative to a given dataset (and not the distribution).


SQ lower bounds are known to imply lower bounds for the local model of privacy with a pure-DP randomizer \cite{KasiviswanathanLNRS11}. This reduction is also based on a simulation of a randomizer using an SQ oracle and, as a result, does not imply any lower bounds for the setting where the number of queries is comparable to $d$. A direct lower bound of $\Omega(d/\eps^2)$ samples for $\eps$-DP local {\em non-interactive} protocols
was given by \cite{duchi2018minimax}. A similar lower bound on non-interactive local-DP protocols was given by \cite{ullman2018tight} for sparse Boolean mean estimation. More recently, a lower bound of $d/\eps^2$ for sparse Gaussian mean estimation in the local DP model that allows interactivity and $(\eps,\delta)$ and $1$-Renyi-DP (or KL-private) randomizers was given by \cite{duchirogers2019lower}.

Several recent works \cite{CheuU21,ChenGKM21,beimel2020round} provide lower bounds for private selection in variants of the multi-message shuffle model. All of these apply only to non-interactive (or one-round shuffling) settings. In fact, \cite{beimel2020round} show that a two-round multi-message protocol can simulate arbitrary secure multi-party protocols. The lower bound of \cite{CheuU21} is the same as ours, $\Omega(\sqrt{d}/\eps)$, and is based on somewhat related ideas. The lower bounds in \cite{ChenGKM21,beimel2020round} are of the form $d/m$, where $m$ is the number of messages sent by each user, and thus are not directly comparable to ours.

The question of implementing selection using simpler APIs has also been raised as an open question for {\em low-sensitivity} queries in~\cite{Steinke25}, and recent work~\cite{LeemanM2025} has shown nearly-optimal algorithms for that model. 
Arbitrary low-sensitivity queries can, however, be challenging to implement, and hence we focus on additive queries in this work.

\section{Preliminaries}
We use the notation $[n]$ for a positive integer $n>0$ to denote the set $\{1,2,\ldots,n\}$. 
For functions $f(n),g(n)$, we sometimes use the notation $f(n)\preceq g(n)$ in conjunction with $f(n)=\O{g(n)}$, i.e., that there exists a constant $C\ge 0$ such that $f(x)\le C\cdot g(x)$ for all $x\ge n$ for a sufficiently large $n$. 

\begin{definition}[Differential privacy~\cite{DworkMNS06j}]
    Let $M$ be a randomized algorithm mapping a dataset $\vx \in \calX^n$ to a range $\mathcal{R}$. We say two datasets $\vx, \vx'$ are neighboring if $x_i = x'_i$ for all but one of the $n$ coordinates. The mechanism $M$ is $(\eps,\delta)$-differentially private if for any pair of neighboring datasets $\vx, \vx'$, and for any event $E \subseteq \mathcal{R}$, 
    \begin{align*}
        \Pr[M(\vx) \in E] \leq e^{\eps}\cdot \Pr[M(\vx') \in E] + \delta.
    \end{align*}
    When an algorithm is $(\eps, 0)$-differentially private, we often call it $\eps$-differentially private.
\end{definition}
\paragraph{Shuffle model.}
In the shuffle model, each individual user first randomizes their own data into one or more messages. 
These messages are then sent to a secure shuffler, which aggregates and randomly permutes all user messages. 
The shuffle model considers an attacker who only sees the collection of messages after the permutation and thus, the set of shuffled messages must satisfy differential privacy. 
Equivalently, an attacker can only observe the histogram of the messages. 
Formally, we have the following definition:

\begin{definition}[Shuffle model of differential privacy]
In the shuffle model of privacy, a protocol $\Pi$ consists of three randomized algorithms:
\begin{itemize}
\item 
A randomizer $\Pi_R:\calX\to\calY^*$ that maps data to vectors, possibly of variable length, which denotes the number of messages sent. 
The protocol is said to be single-message if the probability of sending a single message is $1$ on all inputs; otherwise it is said to be multi-message.
\item 
A shuffler $\Pi_S:\calY^*\to\calY^*$ that applies a uniformly random permutation to all messages. 
\item
An analyzer $\Pi_A:\calY^*\to O$ that performs computation on a set of messages (after permutation).
\end{itemize}
Since the shuffler is the same across all protocols, we use $\Pi=(\Pi_R,\Pi_A)$ to denote each shuffle protocol $\Pi$, so that on an input $x\in\calX^n$, the output is $\Pi(x)=\Pi_A(\Pi_S(\Pi_R(x_1),\ldots,\Pi_R(x_n)))$. 
Hence, we denote the output of the shuffler by $(\Pi_S\circ\Pi_R^n)(x):=\Pi_S(\Pi_R(x_1),\ldots,\Pi_R(x_n))$. 
\end{definition}
We remark on two existing lower bounds for selection in the shuffle model of privacy. 
We first recall the following lower bound by \cite{ChenGKM21}:
\begin{theorem}[Theorem 1.9 in \cite{ChenGKM21}]
\label{lb:shuffle:one}
Any $(\eps,\delta)$-shuffle DP protocol that allows each user to send up to $m$ messages and solves the selection problem on $d$-dimensional inputs with additive error $\frac{1}{10}$ in the average scores, with constant success probability, must use $\Omega\left(\frac{d}{m}\right)$ samples, for any $\eps=\O{1}$ and $\delta=o\left(\frac{1}{d}\right)$.
\end{theorem}

\begin{theorem}[Theorem 7.3 in \cite{CheuU21}]
\label{lb:shuffle:two}
Any $(\eps,\delta)$-shuffle DP protocol that remains private when at least a $\frac{1}{3}$-fraction of users follow the protocol and solves the selection problem on $d$-dimensional inputs with additive error $\alpha$ in the population means and success probability at least $99/100$ must use $\Omega\left(\frac{\sqrt{d}}{\alpha\eps}\right)$ samples for $\delta\log\frac{d}{\delta}\ll\frac{\alpha^2\eps^2}{d}$.
\end{theorem}

\section{Lower Bounds for Private Selection via Noisy Aggregates}
The problem of private selection and its specific instantiations have been studied in a number of distributed models of computation. In this section, we review the known lower bounds that are most relevant to our setting. 
We then give a new lower bound for algorithms that rely on noisy aggregation. 
Specifically, we study algorithms whose \emph{only} access to the dataset is through the noisy aggregation oracle defined in \Cref{def:oracle}, i.e., through noisy values of $\sum_{i\in [n]} \phi(x_i)$. 
Within that constraint, the algorithm is entirely unrestricted: it may issue unboundedly many adaptively chosen queries $\phi: \mathcal{X} \to \mathbb{R}^k$ with an arbitrary output dimension $k$. 
Such queries are referred to as high-dimensional linear or statistical queries. 

Our lower bounds apply to algorithms that receive answers perturbed by standard Gaussian noise. 
This is both the most commonly used method in practice and the method that achieves an asymptotically tight privacy-utility trade-off (for expected squared $L_2$ error). 
Consequently, our lower bounds apply to any algorithm whose privacy guarantee is derived from black-box aggregation guarantees (which encompasses essentially all deployed aggregation APIs), since beating it would require beating the optimal trade-off for aggregation itself.
We further remark that aggregation is often combined with subsampling (using privacy amplification results for subsampling). 
While such results improve the communication and computation efficiency of the implementation, they do not improve the privacy-utility trade-off.



We start by formally defining the Gaussian Aggregation oracle. 
For simplicity of presentation, we assume that each query given to the oracle has $L_2$ sensitivity at most $2$ under replacement adjacency, since its per-record output is in $\calB^k_2(1) = \{x\in \R^k \ |\ \|x\|_2 \leq 1\}$. 
Queries with other bounds on sensitivity can be implemented by scaling the noise parameter of the oracle and then scaling the result back.
\begin{definition}
\label{def:oracle}
For the input dataset $S = (x_1,\ldots,x_n)$, let $\calG(S)$ be the oracle parameterized by $S$ that, given a query $\phi\colon \calX\to \calB^k_2(1)$ and parameter $\sigma > 0$, returns a fresh sample from $\calN(\phi(S), \sigma^2 I_k)$, where $\phi(S) = \sum_{i\in [n]} \phi(x_i)$.
\end{definition}

The privacy properties of a single answer to such a query are well-known. Most conveniently for applications involving composition, a single answer satisfies $2/\sigma^2$-zCDP \cite{bun2016concentrated}.
We assume that privacy bounds for an algorithm with access to the Gaussian Aggregation oracle are derived directly from (advanced) composition of individual privacy guarantees.  Following an analogous notion in \cite{joseph2019role}, we refer to such privacy guarantees as {\em compositional}.
\begin{definition}
Let $\calA$ be an algorithm that is given access to $\calG(S)$ and asks $t$ (adaptively chosen) queries with noise scales $\sigma_1,\ldots,\sigma_t$. 
We say that $\calA$ satisfies {\em compositional} $\rho$-zCDP if the sequence of zCDP parameters of individual queries $\rho_1= 2/\sigma_1^2,\ldots,\rho_t=2/\sigma_t^2$ satisfies $\sum_{i\leq t} \rho_i \leq \rho$.
\end{definition}
We remark that, for simplicity, in this definition the noise scales are chosen non-adaptively, since that is required by the standard composition results for Renyi DP \cite{Mironov17,bun2016concentrated}, and standard applications such as DP-SGD do not need adaptive choice of noise scales. 
However, using fully adaptive composition results for RDP with a Renyi filter \cite{feldman2021individual}, we can also allow the noise scales to be chosen adaptively (with the filter ensuring the condition $\sum_{i\leq t} \rho_i \leq \rho$). 

Our lower bound for selection for $d$ score functions will apply to an easier problem of distinguishing the case where all score functions have mean $1/2$ from the case in which one of the functions has mean $1$. 
Indeed, the index returned by a selection algorithm can be verified using a single aggregation query. 
Formally, we assume that each user $i$ is given a $d$-dimensional Boolean vector $X^i$ and the  $j$-th score function is simply the coordinate $j$ of that vector (denoted by $X^i_j$).
\begin{problem}[Sparse Boolean distinguishing]
\problab{sbm-distinguishing}
For $Z \in \{1,2\}$ and $V\in [d]$ we define a distribution $D_{Z,V}$ over $\{0,1\}^d$. When $Z=2$, $D_{Z,V}$ is uniform on $\{0,1\}^d$. When $Z=1$, $D_{Z,V}$ is uniform on $\{x\in\{0,1\}^d:x_V=1\}$. Each
user $i$ is given an independent sample $X^i$ from $D_{Z,V}$. The goal is to estimate $Z$ for $Z$ chosen uniformly in $\{1,2\}$ and $V$ chosen uniformly in $[d]$.
An algorithm (or protocol) is $2/3$-accurate if it produces an estimate $\widehat Z$ such that $\Pr[\widehat Z=Z]\ge 2/3$, where the probability is under all the random choices of the problem and the algorithm.
\end{problem}

We now state our main lower bound for this problem.
\begin{theorem}
\label{thm:main-lower-bound}
Every $2/3$-accurate algorithm for Sparse Boolean distinguishing (\probref{sbm-distinguishing}) that is given access to $\calG(S)$ and satisfies compositional $\rho$-zCDP must have $n = \Omega(\sqrt{d/\rho})$ samples.
\end{theorem}

The key idea of our proof is to use the fact that the Gaussian Aggregation oracle can be implemented via distributed noise addition \cite{DworkKMMN06}. The transcript of this implementation satisfies privacy guarantees that are exactly $n$ times worse (in terms of zCDP parameters) than the guarantees satisfied by the aggregates since each user adds Gaussian noise at scale $\sigma/\sqrt{n}$. Equivalently, Gaussian aggregation with zCDP parameter $\rho$ can be implemented using a {\em local} $n\rho$-zCDP protocol. By applying zCDP composition to the local zCDP parameters of each query implementation, we get a direct simulation of any algorithm that has access to $\calG$ and satisfies compositional zCDP.
\begin{lemma}[\cite{DworkKMMN06}]
\label{lem:local-zcdp}	
Let $\calA$ be an algorithm that is given access to $\calG(S)$ and satisfies {\em compositional} $\rho$-zCDP. 
Then there exists a local $n\rho$-zCDP protocol $\mbox{Dist-}\calA$ that simulates $\calA$.
\end{lemma}
\begin{proof}
We provide a short, self-contained proof. 
In the distributed noise addition protocol of \cite{DworkKMMN06}, to simulate the Gaussian aggregation oracle, which adds noise $\mathcal{N}(0, \sigma^2 I_k)$, each of the $n$ users independently adds noise $\mathcal{N}(0, (\sigma^2 / n) I_k)$ to their contribution. 
Because the sum of $n$ independent $\mathcal{N}(0, \sigma^2/n)$ random variables is distributed exactly as $\mathcal{N}(0, \sigma^2 I_k)$, the aggregate of the users' messages exactly simulates drawing Gaussian noise $G(S)$. 
Thus for a query with sensitivity $2$, a single user's message with noise variance $\sigma^2/n$ achieves the desired overall noise distribution, and by composition, the protocol is locally $n\rho$-zCDP.
\end{proof}
Therefore it suffices to prove lower bounds for local zCDP protocols. Lower bounds for such protocols that allow adaptive queries were first established by Duchi and Rogers \cite{duchirogers2019lower}, building on the framework of Braverman et al.~\cite{braverman2016communication}. 
In particular, they show that constant-error $1$-sparse Gaussian mean estimation requires $\Omega(d/\rho_{\mbox{local}})$ samples, where $\rho_{\mbox{local}}$ is the RDP parameter of order 1 (or KL) of the protocol.
This result can be used to derive lower bounds for selection by reducing the 1-sparse Gaussian mean estimation to selection. 
Such a reduction will, however, incur an additional logarithmic dependence on $d$ due to the need to bound the range of score functions. 
Our lower bound for Sparse Boolean distinguishing avoids this logarithmic factor, is simpler, and is essentially self-contained. 
We defer the proof, which may be of independent interest, to Appendix~\ref{sec:lower-bound-proof}.

We also note that $\rho$-zCDP implies $(\rho+2\sqrt{\rho\log(1/\delta)},\delta)$-DP for every $\delta>0$ \cite{bun2016concentrated}. Thus, if an algorithm derives an $(\eps,\delta)$-DP guarantee with $\eps\le1$ from its compositional zCDP bound (as is done by the R\'enyi accountants used in practice), then $\rho\le\eps^2/(4\log(1/\delta))$ and \Cref{thm:main-lower-bound} gives $n=\Omega(\sqrt{d\log(1/\delta)}/\eps)$.

\section{The SVT Primitive}
\label{sec:thresh_primitive}
\label{subsec:implementation}
We start by defining the expected behavior of our proposed SVT primitive. Informally, it takes as input a sequence of linear queries and thresholds, and up to a small additive error, the oracle outputs {\sc Yes} if $f(x_1,\ldots,x_n) \geq \tau$, and outputs {\sc No} otherwise. 
Crucially, one initialized stream has a fixed privacy cost, independently of the number of {\sc No} answers. 
To support online and adaptive queries, we allow the queries to be asked one at a time after initializing an SVT stream. 
We discuss both a batch version, which may be simpler to implement (\Cref{sec:svt_batch}), and a more general version that implements the {\sc BetweenThreshold} algorithm of~\cite{BunSU17} (\Cref{sec:svt_between}).
\begin{description}
\item[{\bf Initialize:}] 
Takes as input a privacy parameter $\varepsilon$, and creates the required internal state.
\item[{\bf Query:}] 
Takes as input a function $f : \calX \rightarrow [0,1]$ and a threshold $\tau$, informally asking whether $\sum_{i=1}^n f(x_i) \geq \tau$. Outputs {\sc Yes} or {\sc No}. If it outputs {\sc Yes}, the stream is closed and does not accept any more queries.
\item[{\bf $(\alpha,\beta)$-Accuracy Property:}] 
We say a response is $\alpha$-accurate if the following hold: (a) If for a query $(f, \tau)$, $\sum_{i=1}^n f(x_i) \geq \tau + \alpha$, the output is {\sc Yes}. (b) If for a query $(f, \tau)$, $\sum_{i=1}^n f(x_i) \leq \tau - \alpha$, the output is {\sc No}. 
We say an SVT API is $(\alpha,\beta)$-accurate for a specific query limit if all its responses are $\alpha$-accurate except with probability $\beta$.
\end{description}

Using such a primitive, one can design compositionally differentially private algorithms in a natural way, where the incurred privacy cost is computed by applying a composition theorem to the privacy cost of individual API calls. 
In the upper bounds presented in this work, we use basic composition for pure-DP selection and advanced composition for approximate-DP applications~\cite{DworkRBook}.

One way to implement this primitive is to use the {\sc AboveThreshold} algorithm, which satisfies $\eps$-DP, and gives us $(\alpha, \beta)$-accuracy for any sequence of $K$ queries, with $\alpha = O(\log \frac{2(K+1)}{\beta} / \eps)$. This algorithm is originally due to \cite{DworkNRRV09}, and the formal description below is equivalent to the version in \cite{DworkRBook}.

\begin{algorithm}[!htb]
\caption{SVTInit}
\label{alg:svt_init}
\begin{algorithmic}[1]
\Require{Privacy parameter $\eps$.}
\Ensure{Create an SVT object S}
\State{$S.n \gets \Lap\left(\frac{2}{\eps}\right)$} \Comment{This internal variable is private to the object. }
\State{$S.\eps \gets \eps$}
\State{$S.Done \gets False$.}
\State{Incur privacy cost $\eps$, and return handle to S.}
\end{algorithmic}
\end{algorithm}
\begin{algorithm}[!htb]
\caption{SVTQuery}
\label{alg:svt_query}
\begin{algorithmic}[1]
\Require{Handle S, Function $f : \calX \rightarrow \R$, Threshold $\tau \in \R$. Private input $\vx$}
\Ensure{Output {\sc Yes} or {\sc No}.}
\If{not $S.Done$}
\State{Compute $f(\vx) = \sum_i \texttt{clip}_{[0,1]}(f(x_i))$.}
\State{Sample $u$ from $\Lap\left(\frac{4}{S.\eps}\right)$.}
\If{$f(\vx)+ u > \tau + S.n$}
\State{$S.Done \gets True$}
\State{Return {\sc Yes}}
\Else
\State{Return {\sc No}}
\EndIf
\Else
\State {Return ERROR.}
\EndIf
\end{algorithmic}
\end{algorithm}
The following proofs can be found in~\cite{DworkRBook}.
\begin{theorem}
\label{thm:abovethresh_privacy}
SVT satisfies $\eps$-differential privacy.
\end{theorem}

\begin{theorem}
\label{thm:abovethresh_utility}
For any sequence of functions $f_1,\ldots, f_K:\calX\to[0,1]$, and thresholds $\tau_1,\ldots, \tau_K$, the output of SVT is $(\alpha,\beta)$-accurate for any $\beta\in(0,1)$ and $K\ge 1$, with $\alpha = \frac{8}{\eps}\log\frac{2K}{\beta}$.
\end{theorem}

The clipping step in \Cref{alg:svt_query} ensures that every contribution lies in $[0,1]$ and hence that the privacy guarantee holds for arbitrary functions $f$. 
We state the accuracy guarantee only for functions with range $[0,1]$, for which clipping has no effect.

We describe various design choices for our formalization of the SVT primitive. 
First, our choice of an SVT-based API over a dedicated selection primitive merits justification. 
Our goal is to propose a primitive that is relatively simple to implement while still capturing the power used in most commonly used DP techniques. 
While selection is a more common task, it is not clear that it allows for the implementation of SVT, so SVT is arguably a more general primitive. 
Regardless, abstracting and understanding the relative power of different natural primitives is an important goal in itself. 

Second, as defined, the \textsc{AboveThreshold} algorithm is well-defined for arbitrary sensitivity-1 queries, and indeed the privacy and utility guarantees in \Cref{thm:abovethresh_privacy} and \Cref{thm:abovethresh_utility} hold in that generality. 
Our choice of restricting the functions $f_j(\cdot)$ to be sums $\sum_i f_j(x_i)$ presents a tradeoff. 
On the one hand, we believe that implementation of the primitive in a federated setting would be significantly easier under this additivity restriction. 
On the other hand, we will argue that a large class of differentially private algorithms of interest can be implemented even with this restriction.  
We detail these two arguments in \Cref{subsec:implementation} and \Cref{sec:selection:thresh}.

We can also define a ``batch'' version of the SVT primitive, where a sequence of queries is given in advance and the index of the first {\sc Yes}, if any, is returned. 
This may be simpler to implement, as it does not need to retain internal state between API calls. 
In most settings, the batch and online versions are equivalent, except that the online version may be more communication efficient. 
One could also imagine algorithms that run multiple such SVT primitives in parallel, potentially with different parameters, and where composition could be analyzed using the concurrent composition theorems of~\cite{VadhanW21,Lyu22,VadhanZ23}. 

\textbf{Implementation.}
The implementation of an SVT primitive will be a function of the specific trust model. A few different approaches have been studied for aggregation in the federated setting and we argue that these can accommodate the SVT primitive with relative ease.

The easiest approach is perhaps to use Trusted Execution Environments (TEEs), which have been proposed as an approach to implementing private federated algorithms in~\cite{JinCMRYO24,eichner2024confidential}. 
The simplicity of the SVT primitive makes it fairly well-suited for a TEE implementation, as it may be easy to verify the correctness of an implementation. 
While TEEs can in principle be used to implement general selection or even run an arbitrary central DP algorithm, the value of a minimal primitive like SVT is that it only needs to be audited and analyzed once. 
Running arbitrary analyst code inside the enclave makes the privacy guarantees depend on unaudited code, and both side-channel vulnerabilities and privacy budget state-continuity attacks \cite{JinCMRYO24} worsen as the trusted code base grows. 
In contrast, SVT's trusted state is a single scalar and its logic spans just a few lines of code, which is why we consider it a realistic and highly secure basis for solving more general problems.

A different approach to federated aggregation is to use the two-server or Prio model~\cite{Corrigan-GibbsB17}. 
There is a large body of work on this model, and differentially private aggregation in this setup has been studied by \cite{Talwaretal24}. 
In the two-party settings, noise addition requires some care: there are multiple algorithms for collaborative noise addition, e.g., Franzese et al.~\cite{FranzeseF00JPD25}, though a simple approach is for both servers to independently add scaled noise.
Existing aggregation approaches already allow the two servers to very efficiently construct secret shares of the $f_j(\vx)$ values, as well as validate that each contributed $f_j(x_i)$ is in a bounded range. 
Distributed noise generation for $\gamma$ and $\nu_j$ and the comparison can easily be implemented using garbled circuits and other standard techniques; indeed \cite{BalleBCGKRSS25} implement these ingredients in the same model. 

An alternative approach, potentially combining trust models, is to use a TEE for the comparison part. 
This has the advantage that the input to the TEE, namely the secret shares of $f_j(\vx)$'s, does not scale with the population size $n$. Finally, one may consider an MPC approach to securely computing (shares of) the $f_j(\vx)$ values (e.g. \cite{Bonawitzetal17}). 
There are many approaches to secure noise generation in the MPC model, and the comparison operation can be implemented using general MPC. 
In all of these approaches, we exploit the fact that the computation of $f_j(\vx)$'s can be efficiently implemented using existing aggregation protocols, leaving a constant-size input to the comparison operator rather than one that grows with $n$. 
The noising and comparison operations are now a relatively simple circuit, which can be implemented using standard MPC tools. 

Finally, we remark that SVT is significantly easier to implement securely than a general selection API. 
General selection requires arbitrary score functions as input, whereas each SVT query deals with a single score function at a time. 
This allows for much simpler cryptographic implementations; for example, if garbled circuits were used to implement selection, then \textsf{Argmax} would need to be securely solved over an arbitrary number of candidates. 
Indeed, \cite{DamgardKNOP24} noted that ``the bottleneck in the secure computation protocol is the comparisons required to compute argmax.''
Each candidate set size would require a different circuit, further making verification and auditing harder. 
Moreover, the secure part of the computation invariably adds non-trivial ``constant'' overheads. 
Because SVT processes a single score function at a time and returns a boolean threshold comparison, it allows the analyst to perform the \textsf{Argmax} outside of the cryptographic boundary, resulting in significant computational savings.

\subsection{A Batch Version}
\label{sec:svt_batch}
Here, we define a batch version of the SVT primitive, where one inputs a sequence of queries $(f,\tau)$ with the expectation that the SVT outputs the index of the first query where the answer is above the threshold.
\begin{description}
\item[{\bf Input:}] A sequence of functions $f_1, \ldots, f_K : \calX \rightarrow [0,1]$; a sequence of thresholds $\tau_1, \ldots, \tau_K$; a privacy parameter $\eps$.
\item[{\bf Output:}] An output $o \in [K] \cup \{\bot\}$.
\item[{\bf $(\alpha,\beta)$-Accuracy Property:}] For $j \in [K]$, and the dataset $\vx \in \calX^n$, let  $f_j(\vx) = \sum_{i} \texttt f_j(x_i)$. We say an output $o \in [K] \cup \{\bot\}$ is $\alpha$-accurate on $\vx$ if the following hold.
If $o = \bot$, then $f_j(\vx) \leq \tau_j + \alpha$ for all $j \in [K]$. If $o \in [K]$, then $f_j(\vx) \leq \tau_j + \alpha$ for all $j < o$, and $f_o(\vx) \geq \tau_o - \alpha$. We say an oracle is $(\alpha,\beta)$-accurate if its output is $\alpha$-accurate except with probability $\beta$.
\end{description}

As in the online case, this batch version can be implemented using {\sc AboveThreshold}. 
This batch version has the advantage of not needing to retain internal state between API calls. 

\begin{algorithm}[!htb]
\caption{{\sc AboveThreshold}}
\label{alg:abovethresh}
\begin{algorithmic}[1]
\Require{Dataset $\vx\in\calX^n$. Functions $f_1, \ldots, f_K : \calX \rightarrow [0,1]$; thresholds $\tau_1, \ldots, \tau_K$; privacy parameter $\eps$.}
\Ensure{Output $o \in [K] \cup \{\bot\}$.}
\State{$\gamma \gets \Lap\left(\frac 2 \eps\right)$.}
\For{$j = 1, 2, \ldots, K$}
\State{$f_j(\vx) \gets \sum_i f_j(x_i)$.}
\State{$\nu_j \gets \Lap(\frac 4 \eps)$.}
\If{$f_j(\vx) + \nu_j \geq \tau_j + \gamma$}
\State{Return $j$}
\EndIf
\EndFor
\State{Return $\bot$}
\end{algorithmic}
\end{algorithm}

\subsection{A More Powerful Variant}
\label{sec:svt_between}

Next we present a more general primitive that may be amenable to similar implementation techniques as our simple SVT primitive.

Our more general primitive will allow the query to specify two thresholds $(\tau^L, \tau^H)$ in addition to $f$, and get back information on whether the function $f(\vx)$ is well below $\tau^L$ ($\Low$), well above $\tau^H$ ($\High$), or in between ($\Close$). The primitive can keep answering queries as long as the query answer is not $\Close$. We show an implementation of this primitive below.

\begin{algorithm}[!htb]
\caption{CompareToThresholdsInit}
\label{alg:compare_to_thresh_init}
\begin{algorithmic}[1]
\Require{Privacy parameters $\eps,\delta\in(0,1)$.}
\Ensure{Create a CompareToThresholds Object \textsf{CTT}}
\State{$\textsf{CTT}.n \gets \Lap\left(\frac{2}{\eps}\right)$} \Comment{This internal variable is private to the object. }
\State{$\textsf{CTT}.\eps = \eps$}
\State{$\textsf{CTT}.gap = \frac{12}{\eps}\left(\log \frac {10} \eps + \log \frac 1 \delta + 1\right).$}
\State{$\textsf{CTT}.Done \gets False$.}
\State{Incur privacy cost $(\eps,\delta)$, and return handle to \textsf{CTT}.} 
\end{algorithmic}
\end{algorithm}
\begin{algorithm}[!htb]
\caption{CompareToThresholdsQuery}
\label{alg:compare_to_thresh_query}
\begin{algorithmic}[1]
\Require{Handle \textsf{CTT}, Function $f : \calX \rightarrow \R$, Thresholds $\tau^L, \tau^H \in \R$, private dataset $\vx$}
\Ensure{Output $\Low$, $\High$, or $\Close$.}
\If{not $\textsf{CTT}.Done$ and $\tau^H - \tau^L \geq \textsf{CTT}.gap$}
\State{Compute $f(\vx) = \sum_i \texttt{clip}_{[0,1]}(f(x_i))$.}
\State{Sample $u$ from $\Lap\left(\frac{6}{\textsf{CTT}.\eps}\right)$.}
\If{$f(\vx)+ u > \tau^H - \textsf{CTT}.n$}
\State{Return $\High$}
\EndIf
\If{$f(\vx)+ u < \tau^L + \textsf{CTT}.n$}
\State{Return $\Low$}
\EndIf
\State{$\textsf{CTT}.Done \gets True$}
\State{Return $\Close$}
\Else
\State {Return ERROR.}
\EndIf
\end{algorithmic}
\end{algorithm}
This variant satisfies $(\eps,\delta)$-DP and the privacy bound is proved in \cite{BunSU17}.

\section{Selection in the SVT Model}
\label{sec:selection:thresh}
In this section we give differentially private algorithms for the selection problem in the SVT model. 
We first describe a simple binary search over the range of the scores (\Cref{alg:gen:api}), which uses $O(d\log n)$ SVT queries but loses a factor of $\log n$ in the error. 
We then show that a descending linear scan achieves the optimal error using a single SVT stream but $O(dn)$ queries (\Cref{alg:gen:api:linear}), and that allocating the privacy budget geometrically across the rounds of the binary search removes the $\log n$ factor from the error while keeping $O(d\log n)$ queries (\Cref{alg:geometric_bs}). 
Finally, in \Cref{sec:near-linear:alg} we combine sampling with a filtering pass to obtain the same error using $O(d\log^\ast n+\log^3n)$ queries in expectation (\Cref{alg:filter_fallback}). 
All algorithms access the data only through SVT streams. 
Thus their privacy follows from \Cref{thm:abovethresh_privacy} and basic composition. 
We analyze their accuracy on the event that all SVT answers are accurate.
Throughout this section every score is a sum of per-user values in $[0,1]$, as required by the SVT primitive. Scores with per-user values in $[0,\Delta]$ can be divided by $\Delta$, which multiplies the error bounds by $\Delta$.

\begin{algorithm}[!htb]
\caption{Private binary search}
\label{alg:gen:api}
\begin{algorithmic}[1]
\Require{Oracle $\mathbb{O}$, initial bounds $[\ell,u]=[0,n]$ for non-negative scores, accuracy $\alpha\ge 1$; preserve the last positive index across recursive calls.}
\Ensure{Index $j$ corresponding to approximately highest noisy score $s(v_j)$}
\State{$m\gets\frac{u+\ell}{2}$}
\State{Query oracle $\mathbb{O}$ for index from $1,2,\ldots,d$ with threshold $m$.}
\State{If returned index $j < d$, query oracle $\mathbb{O}$ for index from $j+1,\ldots,d$ with threshold $m$.}
\If{exactly one index in $[d]$ returned} 
\State{Return that index} \Comment{Only one positive answer; its score is within $2\alpha$ of the maximum.}
\ElsIf{two indices returned} 
\State{$\ell\gets m-\alpha$} \Comment{Multiple scores $\geq m$}
\Else
\State{$u\gets m+\alpha$} \Comment{All noisy scores $\leq m$}
\EndIf
\If{$u-\ell > 3\alpha$}
\State{Recurse to find index for upper and lower limits $[\ell,u]$}
\Else
\State{\textbf{Return} the last returned index, or arbitrary index if no indices returned.}
\EndIf
\end{algorithmic}
\end{algorithm}

The algorithm appears in full in \Cref{alg:gen:api} and can be informally interpreted as a binary search over the set of possible scores for the index with the maximum score. 
At a high level, we design a private selection procedure that repeatedly performs noisy threshold comparisons to narrow down the range of possible values for the maximum score, and isolate a high-scoring coordinate. 
Rather than explicitly evaluating all scores with independent noise, as in the na\"{i}ve Laplace mechanism, the algorithm adaptively compares against a threshold using the SVT primitive.
This approach allows us to limit privacy loss to the relatively small number of times the algorithm observes that a score exceeds a threshold, while ignoring the vast majority of coordinates whose scores are clearly suboptimal.

Concretely, the algorithm maintains a candidate interval for the maximum score and performs a logarithmic number of rounds. 
In each round, it picks the mid point of the current interval where the maximum lies and scans the coordinates in their given order using the SVT primitive.
Based on whether zero, one, or multiple coordinates cross the threshold, the algorithm refines its estimate of the maximum and continues. 

If the current width is $W$, the next width is at most $W/2+\alpha$. For $[\ell,u]=[0,n]$ and $\alpha\ge1$, at most $R=\lceil\log_2(2n)\rceil+1$ rounds suffice, each using at most two fresh SVT streams and at most $d$ queries. Allocating $\eps/(2R)$ to each stream makes the algorithm compositionally $\eps$-DP.

Assuming that the SVT oracle is $(\alpha, \beta)$-accurate, all answers in the at most $2R$ streams are correct up to an additive $\alpha$ except with probability $2R\beta$. 
There are three cases. When all $d$ queries are below the current guess $m$, it follows that the true maximum is below $m+\alpha$, so our recursive step is correct. 
When at least two queries are above the current guess $m$, it follows that the maximum is above $m-\alpha$, so once again the maximum score is in $[\ell,u]$. Finally, if there is exactly one score above $m$, the returned score is within $2\alpha$ of the maximum, so we can stop. 
The recursion stops when $u-\ell$ is smaller than $3\alpha$, and outputs the last positive index. Since the maximum score was guaranteed to be in $[\ell,u]$ and the last positive score was at least $\ell$, it follows that the returned index has score within $O(\alpha)$ of the maximum. 
Since our SVT primitive gives $(O(\frac{\log(d/\beta)}{\eps}), \beta)$-accurate answers when called with privacy parameter $\eps$, we get the following bound.
\begin{theorem}
For $d\ge2$, $n\ge1$, $0<\eps\le1$, and $\beta\in(0,1)$, \Cref{alg:gen:api}, with privacy budget $\eps/(2R)$ for each SVT stream, is $\eps$-DP and has additive error $O(\log(2n)\log(2d\log(2n)/\beta)/\eps)$ with probability at least $1-\beta$, using $O(d\log(2n))$ queries.
\end{theorem}

\begin{corollary}
Given an oracle $\mathbb{O}$ that is $(\alpha,\beta)$-accurate for $\alpha=\O{\frac{\log d}{\eps}}$ and $\beta=\frac{1}{\poly(d)}$, there exists an $\eps$-DP algorithm that performs approximate maximization for decomposable score functions. 
\end{corollary}

Unfortunately, the additive error of this binary search approach is $\O{\frac{\log d\log n}{\eps}}$ rather than the optimal error of $\O{\frac{\log d}{\eps}}$ achievable in the central setting. 
This is because the binary search procedure may need to split the privacy budget small enough to invoke the SVT mechanism $O(\log n)$ times. 
Thus one may ask whether it is possible to achieve the optimal additive error of $\O{\frac{\log d}{\eps}}$. 

\subsection{Linear Scan}
We first observe that a simple linear scan can achieve an optimal additive error. 
Namely, the protocol simply starts from the maximum possible score of an index and repeatedly calls the SVT API. 
For each query where the SVT API does not return a coordinate, the protocol decreases the guess for the maximum possible score of an index, repeating until the guess is within the typical bounds for Laplacian noise added to the score of the maximum coordinate. 
Since this approach will report only a single coordinate from the SVT API, it spends less privacy budget than the binary search approach. 
Assuming the score for each user is normalized to $[0,1]$, we have the following guarantees:
\begin{restatable}{theorem}{thmselectionlinearscan}
For $n\ge1$, $d\ge2$, and $\beta\in(0,1)$, \Cref{alg:gen:api:linear} with $[\ell,u]=[0,n]$ and a single $\eps$-DP SVT stream is an $\eps$-DP selection algorithm with additive error $O(1+\log(d(n+1)/\beta)/\eps)$ with probability at least $1-\beta$, using $O(d(n+1))$ queries.
\end{restatable}
\begin{proof}
Privacy follows from the fact that we run exactly one SVT oracle, and stop once we get one score above the threshold. The accuracy bound follows from the fact that we use decreasing thresholds and the SVT oracle is $\alpha$-accurate. 
Finally, note that since the score for each user is in $[0,1]$, then the score for each index is contained within the interval $[0,n]$. 
Hence, there are $n+1$ threshold levels when scanning the integer bounds from $n$ through $0$, each across $d$ coordinates, for a total of $\O{nd}$ queries. 
\end{proof}

\begin{algorithm}[!htb]
\caption{Private linear scan}
\label{alg:gen:api:linear}
\begin{algorithmic}[1]
\Require{Oracle $\mathbb{O}$, upper limit $u$, lower limit $\ell$ on score}
\Ensure{Index $j\in[d]$ with approximately maximum true score}
\For{$m \gets u, u-1, u-2, \ldots, \ell$}
\State{Query the same initialized SVT stream $\mathbb{O}$ at threshold $m$ for all indices  $j\in[d]$}
\If{an index $j\in[d]$ is returned}
\State{\textbf{Return} $j$}
\Comment{End procedure}
\EndIf
\EndFor
\State{\textbf{Return} an arbitrary index in $[d]$}
\end{algorithmic}
\end{algorithm}

\subsection{Improved Binary Search}

The downside is that to achieve the additive error $O(1+\log(d(n+1)/\beta)/\eps)$, the linear scan requires using $\O{nd}$ queries.
By comparison, binary search requires $\O{d\log n}$ queries but has additive error $\O{{\log d\log n}/{\eps}}$. 
We now show that allocating the privacy budget geometrically across the rounds of the binary search removes the $\log n$ factor from the error while keeping $O(d\log n)$ queries. Indeed, the error introduced in a round of the binary search is halved in each of the subsequent rounds. Therefore the early rounds can use a much smaller privacy budget than the last ones. Specifically, in \Cref{alg:geometric_bs} round $t$ uses budget proportional to $\gamma^{T-t}$ for $\gamma=3/4$. Thus its contribution to the final error is proportional to $(2\gamma)^{-(T-t)}$ and the total error is dominated by the last round. We state the algorithm for an arbitrary set of candidates $P$, an upper bound $u$ on the maximum score in $P$ and a witness index whose score is at least a given lower bound $\ell$, since we use it as a subroutine in \Cref{sec:near-linear:alg} and in \Cref{sec:applications}. 

\begin{algorithm}[!htb]
\caption{Private geometric binary search (\textproc{GeometricBinarySearch})}
\label{alg:geometric_bs}
\begin{algorithmic}[1]
\Require{Oracle $\mathbb{O}$, candidate set $P\subseteq[d]$, interval $[\ell,u]$ with $\max_{j\in P}s_j(D)\le u$, witness $j_0\in[d]$ with $\ell\le s_{j_0}(D)$, number of rounds $T\ge1$, privacy budget $\eps>0$, failure probability $\beta\in(0,1)$}
\Ensure{Index $j^\star$ and an interval $[\ell,u]$ such that $s_{j^\star}(D)\ge \ell$ and $\max_{j\in P}s_j(D)\le u$}
\State{$\gamma\gets 3/4$, $Z\gets \frac{1-\gamma^T}{1-\gamma}$, $j^\star\gets j_0$}
\For{$t=1,\dots,T$}
\State{$\eps_t\gets \frac{\eps}{Z}\gamma^{T-t}$}
\State{$\alpha_t\gets \frac{32}{\eps_t}\log\left(\frac{2|P|T}{\beta}\right)$}
\State{$m\gets \frac{\ell+u}{2}$}
\State{Query oracle $\mathbb{O}$ on all candidates in $P$ with threshold $m$, budget $\eps_t$, and failure probability $\beta/T$}
\If{no index returned}
\State{$u\gets \min\{u,m+\alpha_t\}$} \Comment{All noisy scores $\lesssim m$}
\Else{(oracle returns index $j$)}
\If{$m-\alpha_t>\ell$}
\State{$\ell\gets m-\alpha_t$}
\State{$j^\star\gets j$}
\EndIf
\EndIf
\EndFor
\State{\textbf{Return} $(j^\star,\ell,u)$}
\end{algorithmic}
\end{algorithm}

We next analyze \Cref{alg:geometric_bs}. 
Let $\calE$ denote the event that every oracle call satisfies its accuracy guarantee. 
Each of the $T$ calls fails with probability at most $\frac{\beta}{T}$, and thus $\PPr{\calE}\ge 1-\beta$ by a union bound.
We thus condition on the event $\calE$ for the remainder of the utility analysis. 
\begin{lemma}
For all $t\in[T+1]$, we have $s_{j^\star}(D)\ge \ell_t$ and $S^\star\le u_t$, where $S^\star=\max_{j\in P} s_j(D)$.
\end{lemma}
\begin{proof}
We prove the claim by induction on $t$. 
For $t=1$, we have $s_{j^\star}(D)=s_{j_0}(D)\ge\ell=\ell_1$ and $S^\star\le u=u_1$ by the assumptions on the input.
Suppose the inductive hypothesis holds at step $t$ and let $m_t=(\ell_t+u_t)/2$. 
If no index is returned, then the oracle accuracy gives $S^\star\le m_t+\alpha_t$ and thus 
\[S^\star\le \min\{u_t,m_t+\alpha_t\}=u_{t+1}.\]
We also have $\ell_{t+1}=\ell_t$.

On the other hand, if index $j$ is returned, then $s_j(D)\ge m_t-\alpha_t$. 
If $m_t-\alpha_t>\ell_t$, the algorithm sets $\ell_{t+1}=m_t-\alpha_t$ and $j^\star=j$, so $s_{j^\star}(D)\ge \ell_{t+1}$. 
Otherwise, we again have $\ell_{t+1}=\ell_t$. 
Moreover, in both cases, we have $u_{t+1}=u_t$.
Therefore, the inductive step is complete. 
\end{proof}
We next bound the size of the search space, showing that it geometrically decreases at each step $t\in[T]$. 
\begin{lemma}
\label{lem:width:bound}
Let $W_t=u_t-\ell_t$. 
Then for all $t\le T$, we have $W_{t+1}\le \frac{W_t}{2}+\alpha_t$. 
\end{lemma}
\begin{proof}
Since $m_t=\frac{\ell_t+u_t}{2}$, then we have
\[m_t-\ell_t=u_t-m_t=\frac{W_t}{2}.\]
Now, if no index is returned, then we have
\[W_{t+1}=u_{t+1}-\ell_t \le (m_t+\alpha_t)-\ell_t = \frac{W_t}{2}+\alpha_t.\]
Otherwise, if an index is returned, we have
\[W_{t+1}=u_t-\ell_{t+1} \le u_t-(m_t-\alpha_t) = \frac{W_t}{2}+\alpha_t.\]
\end{proof}

The following theorem summarizes the guarantees of \Cref{alg:geometric_bs}.
\begin{theorem}
\label{thm:gbs}
\Cref{alg:geometric_bs} is $\eps$-differentially private and makes at most $|P|T$ SVT queries. Let $W=u-\ell$. With probability at least $1-\beta$, it returns $(j^\star,\ell',u')$ such that $s_{j^\star}(D)\ge\ell'$, $\max_{j\in P}s_j(D)\le u'$, and
\[u'-\ell'\le \frac{W}{2^T}+\frac{384}{\eps}\log\left(\frac{2|P|T}{\beta}\right).\]
In particular, for $P=[d]$, $[\ell,u]=[0,n]$, $j_0=1$ and $T=\max\{1,\lceil\log_2 n\rceil\}$, with probability at least $1-\beta$ the returned index satisfies
\[S^\star-s_{j^\star}(D)\le 1+\frac{384}{\eps}\log\left(\frac{2d(1+\log_2 n)}{\beta}\right).\]
\end{theorem}
\begin{proof}
\Cref{alg:geometric_bs} accesses the data only through $T$ oracle calls, and the call in round $t$ has privacy budget $\eps_t$ and answers at most $|P|$ queries. By the choice of $Z$, $\sum_{t\in[T]}\eps_t=\eps$, and therefore privacy follows from basic composition.
Conditioned on the event $\calE$, we have
\[S^\star - s_{j^\star}(D) \le u_{T+1} - \ell_{T+1} = W_{T+1}.\]
By \Cref{lem:width:bound}, we have
\[W_{T+1} \le \frac{W_1}{2^T} + \sum_{t=1}^T \frac{\alpha_t}{2^{T-t}}.\]
Since $\alpha_t=\frac{32}{\eps_t}\log\left(\frac{2|P|T}{\beta}\right)=\frac{32Z}{\eps}\gamma^{-(T-t)}\log\left(\frac{2|P|T}{\beta}\right)$, we have
\[\sum_{t=1}^T \frac{\alpha_t}{2^{T-t}} = \frac{32Z}{\eps}\log\left(\frac{2|P|T}{\beta}\right)\sum_{k=0}^{T-1}(2\gamma)^{-k}\le \frac{384}{\eps}\log\left(\frac{2|P|T}{\beta}\right),\]
where we used that $Z\le\frac{1}{1-\gamma}=4$ and $\sum_{k\ge0}(2\gamma)^{-k}=3$. This gives the first claim. For the second claim we note that $W_1=n\le 2^T$ and $T\le 1+\log_2 n$.
\end{proof}

\begin{corollary}
If $d\ge2$, $0<\eps\le1$, $\beta = d^{-c}$ for a fixed $c>0$, and $n\le e^{\poly(d)}$, then \Cref{alg:geometric_bs} with $P=[d]$, $[\ell,u]=[0,n]$, $j_0=1$ and $T=\max\{1,\lceil\log_2 n\rceil\}$ makes $O(d\log n)$ SVT queries and, with probability at least $1-\beta$, returns $j^\star$ such that
\[S^\star - s_{j^\star}(D) = \O{\frac{\log d}{\eps}}.\]
\end{corollary}

Running \Cref{alg:geometric_bs} on a single candidate gives a simple way to privately estimate a sum using the SVT API, which we use in \Cref{sec:applications}.
\begin{corollary}
\label{cor:gbs:estimate}
Let $f:\calX\to[0,1]$, $\eps>0$ and $\beta\in(0,1)$, and let $T=\max\{1,\lceil\log_2 n\rceil\}$. There is an $\eps$-differentially private algorithm that makes $T$ SVT queries and, with probability at least $1-\beta$, outputs an interval of length at most $1+\frac{384}{\eps}\log\frac{2T}{\beta}$ that contains $\sum_{i\in[n]}f(x_i)$.
\end{corollary}
\begin{proof}
We run \Cref{alg:geometric_bs} with the single candidate $f$, the interval $[0,n]$, the witness $f$ (whose score is nonnegative) and $T$ rounds. By \Cref{thm:gbs}, the returned interval $[\ell',u']$ contains $\sum_{i\in[n]}f(x_i)$ and has length at most $n/2^T+\frac{384}{\eps}\log\frac{2T}{\beta}$.
\end{proof}

\subsection{Near-Linear Algorithm}
\label{sec:near-linear:alg}
We now describe an algorithm that achieves the error of \Cref{alg:geometric_bs} using only $O(d\log^\ast n+\log^3 n)$ SVT queries in expectation. For simplicity of notation we write $s_j$ for $s_j(D)$ and $S^\star=\max_{j\in[d]}s_j$.

Like \Cref{alg:geometric_bs}, our algorithm maintains an interval $[\ell,u]$ that contains $S^\star$ together with a witness index $j_w$ whose score is at least $\ell$. \Cref{alg:geometric_bs} halves the length of this interval in each round and queries all $d$ candidates in every round. Our algorithm instead proceeds in {\em phases}, each of which reduces the length of the interval from $x\alpha$ to $O(\alpha\log^3x)$ using $O(d+\log^3x)$ queries in expectation, where $\alpha$ is (up to a constant factor) the target error. Since the length decreases so rapidly, $O(\log^\ast n)$ phases suffice to reduce it to $O(\alpha)$, and a final run of \Cref{alg:geometric_bs} with a constant number of rounds completes the selection.

The main step of a phase is a {\em filtering pass} in which we query every candidate once against a fixed threshold $\tau$. Every candidate that answers {\sc No} has score at most $\tau$ (up to the accuracy of the SVT). Thus it remains to search the set $Y$ of candidates that answer {\sc Yes}, which is cheap when $Y$ is small. At the same time, each {\sc Yes} answer closes the SVT stream and therefore the privacy cost of a filtering pass grows linearly with $|Y|$. Thus we need a threshold that only few candidates exceed. To construct this threshold, we first form a random sample $P$ by including each candidate independently with probability $1/t$, where $t=\lceil\log_2x\rceil$. We estimate the maximum score in $P$ by running \Cref{alg:geometric_bs} on $P$ with $t$ rounds, which makes at most $|P|t$ queries, or about $d$ queries in expectation. We then set $\tau$ slightly above this estimate, with enough slack to account for both the estimation error and the SVT error, and run the filtering pass. When all SVT answers are accurate, only candidates whose scores exceed the maximum score in $P$ can answer {\sc Yes}. The number of such candidates is dominated by a geometric random variable with mean $t$ and thus it is less than $t^2$ except with probability $e^{-t}$. When $|Y|<t^2$, we search $Y$ using \Cref{alg:geometric_bs} with $t$ rounds, which makes at most $t^3$ queries.

We use privacy budget $\eps/t^3$ for every stream of the filtering pass and for each of the two runs of \Cref{alg:geometric_bs} in a phase. Therefore the accuracy of the SVT answers, and also the length of the new interval, is $O(\alpha t^3)=O(\alpha\log^3x)$. The filtering pass initializes at most $t^2$ streams and thus a phase uses budget $O(\eps/t)$ in total. This budget at least doubles from one phase to the next and thus the total budget of all phases is dominated by the last one. Two additional details ensure that the privacy guarantee holds on every execution, and not only when all SVT answers are accurate. First, the parameters of every phase are fixed in advance and do not depend on the data. Second, the filtering pass stops after $t^2$ {\sc Yes} answers and we then search all $d$ candidates instead of $Y$. This makes $dt$ queries but happens with probability at most $e^{-t}$ when all SVT answers are accurate. We present the full details in \Cref{alg:filter_fallback}. There, $\textproc{GeometricBinarySearch}(P,[\ell,u],j,T,\eps,\beta)$ denotes a run of \Cref{alg:geometric_bs} on the candidate set $P$ with the interval $[\ell,u]$, the witness $j$, $T$ rounds, privacy budget $\eps$ and failure probability $\beta$.

\begin{algorithm}[!htb]
\caption{Private selection by sampling and filtering}
\label{alg:filter_fallback}
\begin{algorithmic}[1]
\Require{Oracle $\mathbb{O}$, candidates $[d]$ with scores in $[0,n]$, privacy budget $\eps\in(0,1]$, failure probability $\beta\in(0,1)$}
\Ensure{Index $j^\star$ with approximately maximum score}
\State{$N\gets3\log_2^2(2n)$, $\beta'\gets\beta/N$, $\alpha\gets\frac{8}{\eps}\log\frac{384d\log_2^3(2n)}{\beta}$}
\State{$\ell\gets0$, $u\gets n$, $j_w\gets1$, $x\gets n/\alpha$}
\While{$x\ge2^{64}$}\Comment{one phase}
\State{$t\gets\lceil\log_2x\rceil$, $\eps'\gets\eps/t^3$, $a\gets t^3\alpha$}
\State{Sample $P'\subseteq[d]$ by including each candidate independently with probability $1/t$, and let $P\gets P'\cup\{j_w\}$}
\State{$(j_P,\ell_P,\cdot)\gets\textproc{GeometricBinarySearch}(P,[\ell,u],j_w,t,\eps',\beta')$}
\State{$\tau\gets\ell_P+50a$, $Y\gets\emptyset$, and initialize an SVT stream with budget $\eps'$} \Comment{filtering pass}
\For{$j=1,\ldots,d$}
\If{the query $(f_j,\tau)$ to the current stream returns {\sc Yes}}
\State{$Y\gets Y\cup\{j\}$}
\State{\textbf{if} $|Y|=t^2$ \textbf{then break}, \textbf{else} initialize a new SVT stream with budget $\eps'$}
\EndIf
\EndFor
\State{$C\gets[d]$ if $|Y|=t^2$, and $C\gets Y\cup\{j_P\}$ otherwise}
\State{$(j_w,\ell,u')\gets\textproc{GeometricBinarySearch}(C,[\ell_P,u],j_P,t,\eps',\beta')$}
\State{$u\gets\min\{u,\max\{u',\tau+a\}\}$, $x\gets64t^3$}
\EndWhile
\State{$(j^\star,\cdot,\cdot)\gets\textproc{GeometricBinarySearch}([d],[\ell,u],j_w,64,\eps/2,\beta')$}
\State{\Return $j^\star$}
\end{algorithmic}
\end{algorithm}

We now analyze \Cref{alg:filter_fallback}. Let $k$ denote the number of phases and, for $i<k$, let $x_i$ and $t_i=\lceil\log_2x_i\rceil$ denote the values of $x$ and $t$ in phase $i$ (starting with $i=0$). We also denote by $x_k$ the value of $x$ after the last phase. Note that these values depend only on $n$, $d$, $\eps$ and $\beta$. By definition, $x_{i+1}=64t_i^3$ and $t_i\ge64$ for every $i<k$. This implies that $t_{i+1}\le7+3\log_2t_i\le t_i/2$. Therefore
\begin{equation}
\sum_{i<k}\frac{1}{t_i}\le\frac{2}{t_{k-1}}\le\frac{1}{32}.
\label{nl:sum}
\end{equation}
Further, it is easy to see that two consecutive phases reduce any sufficiently large value of $x$ to at most $\log_2x$ and thus $k=O(\log^\ast n)$.

We first show that the algorithm is private, which follows from basic composition since the budgets of all phases are fixed in advance.
\begin{lemma}
\label{lem:priv:filter_fallback}
\Cref{alg:filter_fallback} is $\eps$-differentially private.
\end{lemma}
\begin{proof}
The algorithm accesses the data only through SVT streams, either in a filtering pass or in a run of \Cref{alg:geometric_bs}. In phase $i$ each of the two runs of \Cref{alg:geometric_bs} uses budget $\eps/t_i^3$, and the filtering pass initializes at most $t_i^2$ streams with budget $\eps/t_i^3$ each. Thus phase $i$ uses budget at most $(t_i^2+2)\eps/t_i^3\le2\eps/t_i$. By Equation~\eqref{nl:sum}, all phases together use budget at most $\eps/16$. The final run of \Cref{alg:geometric_bs} uses budget $\eps/2$ and therefore the claim follows from basic composition.
\end{proof}

Next we analyze the accuracy of the algorithm. Every run of \Cref{alg:geometric_bs} is made with failure probability $\beta'=\beta/N$. Every stream of a filtering pass answers at most $d$ queries and thus, by \Cref{thm:abovethresh_utility}, it is $\frac{8}{\eps'}\log\frac{2dN}{\beta}$-accurate except with probability $\beta/N$. The total number of runs of \Cref{alg:geometric_bs} and streams of filtering passes is at most $1+\sum_{i<k}(t_i^2+2)\le N$, since $t_{i+1}\le t_i/2$ and $t_0\le\log_2(2n)$ (as $\alpha\ge1$). Let $\calE$ denote the event that every run of \Cref{alg:geometric_bs} satisfies the guarantee of \Cref{thm:gbs} and every stream of a filtering pass is accurate as above. By a union bound, $\Pr[\calE]\ge1-\beta$. By our choice of $\alpha$, on the event $\calE$ every stream of the filtering pass in phase $i$ is $t_i^3\alpha$-accurate, and every run of \Cref{alg:geometric_bs} in phase $i$ with input interval $[\ell',u']$ returns an interval of length at most $(u'-\ell')/2^{t_i}+48t_i^3\alpha$. Similarly, the final run returns an interval of length at most $(u-\ell)/2^{64}+96\alpha$. 

We now show that on the event $\calE$ every phase maintains the interval and the witness, and reduces the length of the interval as intended. The proof formalizes the description above. The threshold of the filtering pass exceeds the maximum score in $P$ by more than the accuracy of the stream. Thus only candidates whose scores exceed this maximum can answer {\sc Yes}. Every candidate that answers {\sc No} is certified to have score at most $\tau+a$, and the run of \Cref{alg:geometric_bs} on $C$ certifies the scores of the remaining candidates. 

\begin{lemma}
\label{lem:util:filter_fallback}
On the event $\calE$, for every $i\le k$, at the beginning of phase $i$ (or of the final run of \Cref{alg:geometric_bs} if $i=k$) we have $S^\star\le u$, $s_{j_w}\ge\ell$ and $u-\ell\le x_i\alpha$. Moreover, the filtering pass of phase $i<k$ stops after $t_i^2$ {\sc Yes} answers only if at least $t_i^2$ candidates have score larger than $\max_{j\in P}s_j$.
\end{lemma}
\begin{proof}
We prove the claim by induction on $i$. For $i=0$ it holds since all scores lie in $[0,n]$ and $x_0=n/\alpha$. Now consider phase $i<k$ and let $t=t_i$ and $a=t^3\alpha$. 
By the inductive hypothesis, the run of \Cref{alg:geometric_bs} on $P$ returns $\ell_P\ge\ell$ (the lower end of the interval never decreases in \Cref{alg:geometric_bs}) such that $s_{j_P}\ge\ell_P$ and $S^\star_P=\max_{j\in P}s_j\le\ell_P+(u-\ell)/2^t+48a\le\ell_P+49a$. 
Therefore every candidate $j$ with $s_j\le S^\star_P$ satisfies $s_j\le\tau-a$ and answers {\sc No} in the filtering pass. 
This proves the second claim since every candidate is queried at most once. 
Further, every candidate $j\notin Y$ that is queried in the filtering pass answers {\sc No} and thus satisfies $s_j\le\tau+a$.

Let $(j',\ell',u')$ denote the output of the run of \Cref{alg:geometric_bs} on $C$, so that the new values are $j_w=j'$, $\ell=\ell'$ and $u=\min\{u,\max\{u',\tau+a\}\}$. 
This run starts from the interval $[\ell_P,u]$ with the witness $j_P$ and thus $s_{j'}\ge\ell'$, $S^\star_C=\max_{j\in C}s_j\le u'$ and $u'-\ell'\le49a$. 
If the filtering pass does not stop early then every candidate outside of $C$ answered {\sc No} and thus $S^\star\le\max\{u',\tau+a\}$. Otherwise $C=[d]$ and $S^\star=S^\star_C\le u'$. In both cases the new value of $u$ is at least $S^\star$. 
The length of the new interval is at most $\max\{u'-\ell',\tau+a-\ell'\}$. 
If $Y\ne\emptyset$ then some candidate in $C$ answered {\sc Yes}. 
Thus $S^\star_C\ge\tau-a$ and $\ell'\ge S^\star_C-49a\ge\tau-50a$. 
If $Y=\emptyset$ then $\ell'\ge\ell_P=\tau-50a$. 
In both cases the new interval has length at most $51a\le x_{i+1}\alpha$.
\end{proof}

Finally, we bound the expected number of queries. 
The run of \Cref{alg:geometric_bs} on $P$ makes about $d$ queries in expectation and the filtering pass makes at most $d$ queries. 
The run on $C$ makes at most $t^3$ queries, unless the filtering pass stops early. 
By \Cref{lem:util:filter_fallback}, on the event $\calE$ this happens only if none of the $t^2$ candidates with the largest scores is in the sample, which has probability at most $e^{-t}$.
\begin{lemma}
\label{lem:queries:filter_fallback}
If $\beta\le1/d$ then the expected number of SVT queries made by \Cref{alg:filter_fallback} is at most $3dk+2\log_2^3(2n)+64d=O(d\log^\ast n+\log^3n)$.
\end{lemma}
\begin{proof}
Consider phase $i<k$ and let $t=t_i$. The run of \Cref{alg:geometric_bs} on $P$ makes $|P|t$ queries and $\E[|P|]\le1+d/t$. The filtering pass makes at most $d$ queries. The run on $C$ makes at most $t^3$ queries if the filtering pass does not stop early, and $dt$ queries otherwise. Let $j_1,\ldots,j_d$ denote the candidates sorted by their scores in decreasing order (with ties broken arbitrarily). If $P'$ includes one of $j_1,\ldots,j_{t^2}$ then fewer than $t^2$ candidates have score larger than $\max_{j\in P}s_j$. By \Cref{lem:util:filter_fallback}, this implies that the filtering pass stops early only if $\calE$ does not hold or $P'$ includes none of $j_1,\ldots,j_{t^2}$, which happens with probability at most $\beta+(1-1/t)^{t^2}\le\beta+e^{-t}$. Therefore the expected number of queries in phase $i$ is at most
\[(1+d/t)t+d+t^3+(\beta+e^{-t})dt\le3d+t^3+2t,\]
where we used that $\beta d\le1$ and $te^{-t}\le1$. Summing over the phases and using that $t_{i+1}\le t_i/2$ and $t_0\le\log_2(2n)$, we get that all phases together make at most $3dk+2\log_2^3(2n)$ queries in expectation. The final run of \Cref{alg:geometric_bs} makes $64d$ queries.
\end{proof}

By \Cref{lem:util:filter_fallback}, on the event $\calE$ the final run of \Cref{alg:geometric_bs} starts from an interval of length $u-\ell\le x_k\alpha<2^{64}\alpha$ that contains $S^\star$, with a witness whose score is at least $\ell$. Therefore it returns $j^\star$ such that $S^\star-s_{j^\star}\le\alpha+96\alpha=97\alpha$. Together with \Cref{lem:priv:filter_fallback,lem:queries:filter_fallback}, this gives the following more general version of \Cref{thm:select:thresh}.
\begin{restatable}{theorem}{thmlogstar}
\label{thm:logstar}
Let $n\ge2$, $d\ge1$, $\eps\in(0,1]$ and $\beta\in(0,1/d]$. \Cref{alg:filter_fallback} is $\eps$-differentially private and, with probability at least $1-\beta$, it outputs an index $j^\star$ such that
\[\max_{j\in[d]}s_j(D)-s_{j^\star}(D)\le\frac{776}{\eps}\log\frac{384d\log_2^3(2n)}{\beta}=O\left(\frac{1}{\eps}\log\frac{d\log n}{\beta}\right).\]
The expected number of SVT queries made by the algorithm is $O(d\log^\ast n+\log^3n)$. In particular, for $\beta=1/(dn)$ the expected additive error is $O\left(\frac{\log(dn)}{\eps}\right)$, and if $d\ge\log^3n$ then the expected number of queries is $O(d\log^\ast d)$.
\end{restatable}

\section{Applications in the SVT Model}
\label{sec:applications}
Because private selection appears in many settings, \Cref{thm:select:thresh} has a wide range of applications. 
In this section, we present several examples demonstrating how the SVT primitive substantially improves upon existing models.

\paragraph{$F_p$ losses.}
We first note that our framework extends naturally to the setting where the score function of a coordinate $i\in[d]$ is either the $F_p$ moment or the $L_p$ norm of the individual components. 
For $p\ge1$, in the first case we have $s(i)=\sum_{j\in[n]}|v_i^{(j)}|^p$, and in the second case we have $s(i)=\left(\sum_{j\in[n]}|v_i^{(j)}|^p\right)^{1/p}$.
Observe that each user $j\in[n]$ can compute the vector $w^{(j)}$, defined so that $w^{(j)}_i=|v^{(j)}_i|^p$ for all $i\in[d]$.
Then the $F_p$ score is $s(i)=\sum_{j\in[n]}w_i^{(j)}$. If all the residuals satisfy $|v_i^{(j)}|\le b$, then every $w_i^{(j)}$ lies in $[0,b^p]$ and the users can run our selection algorithm on the vectors $w^{(1)}/b^p,\ldots,w^{(n)}/b^p$, thereby achieving additive error $O\left(\frac{b^p\log d}{\eps}\right)$ for the $F_p$ score (when $n$ is polynomial in $d$). Since $(y-a)^{1/p}\ge y^{1/p}-a^{1/p}$ for all $0\le a\le y$, an additive error $a$ for the $F_p$ score implies additive error at most $a^{1/p}$ for the $L_p$ score.
This formulation captures a number of standard loss measures, including:
\begin{itemize}
\item 
\emph{Squared Euclidean loss, i.e., $F_2$ loss:}  
If $v_i^{(j)}$ represents a residual or prediction error, then $s(i)=\sum_{j\in[n]}(v_i^{(j)})^2$ measures the total squared error at coordinate $i$. 
As the score function is the squared $\ell_2$ norm of the vector $(v_i^{(1)},\ldots,v_i^{(n)})$, it penalizes large deviations more heavily than small ones and thus selection based on $s(i)$ identifies coordinates where the error vector has large length. 
This loss is widely used in regression, least-squares estimation, and many machine learning tasks. 
\item 
\emph{Absolute deviation, i.e., $F_1$ or $L_1$ loss:}  
Taking $p=1$ results in the score function $s(i)=\sum_{j\in[n]}|v_i^{(j)}|$, a standard robust alternative to squared loss. 
The $L_1$ loss is less sensitive to outliers and is commonly used in robust estimation and sparse recovery when occasional large deviations are expected.
\item 
\emph{Higher-order moments, i.e., $F_p$ loss for $p>2$:}  
For larger $p$, the score $s(i)=\sum_{j\in[n]}|v_i^{(j)}|^p$ emphasizes rare but large deviations. 
The corresponding $L_p$ score $s(i)=\left(\sum_{j\in[n]}|v_i^{(j)}|^p\right)^{1/p}$ serves as a smooth proxy for the $\ell_\infty$ loss $s_\infty(i)=\max_{j\in[n]}|v_i^{(j)}|$, which captures worst-case error. 
This is useful in applications requiring uniform accuracy or robustness guarantees, while using finite $p$ retains computational feasibility in many applications.  
For instance, in distributed or federated learning, one might want to identify coordinates where the largest per-user errors occur, ensuring that no single coordinate becomes a bottleneck or source of extreme error.  
Similarly, in sensor networks or traffic monitoring, large $F_p$ or $L_p$ scores correspond to coordinates with unusually high variability or spikes, enabling anomaly detection and more robust decision-making.
\end{itemize}

\paragraph{Heavy-hitters with categories.}
In many real-world datasets, items are associated with multiple categories, forming sparse binary vectors where each item is represented by a vector indicating its presence across various categories. 
The goal is to identify the top $k$ categories that contain the most items. 
This problem arises in a wide range of applications. 
For instance, in e-commerce, determining the most popular product categories can inform inventory management and marketing strategies. 
In social media, identifying trending topics or hashtags helps improve content curation and user engagement. 
Similarly, in healthcare, analyzing patient data to find prevalent disease categories can aid in resource allocation and treatment planning. 
Efficiently solving this problem is particularly relevant in these applications where both accuracy and privacy are important.  

Formally, suppose there are $n$ items, which can each be mapped to up to $r$ categories in a universe of $d$ categories.  
In other words, there are $n$ vectors that are each $r$-sparse binary vectors $v^{(1)},\ldots,v^{(n)}\in\{0,1\}^d$ such that $\|v^{(i)}\|_0\le r$ for all $i\in[n]$.
The goal is to find the top $k$ categories with the largest number of items. 
Specifically, the goal is to find the largest $k$ coordinates of the vector $v=\sum_{i\in[n]}v^{(i)}\in\mathbb{R}^d$. 

We first observe that because each item can be in up to $r$ categories, the $L_1$ sensitivity of the count vector $v$ is $2r$ (when one item is replaced). Thus the standard approach of releasing all $d$ counts with Laplace noise calibrated to this sensitivity adds noise of scale $2r/\eps$ to every count, and the resulting error in the largest counts is $\Theta(r\log d/\eps)$ with constant probability. On the other hand, each individual count changes by at most $1$ when an item is replaced. Therefore we can apply our selection algorithm directly to the score functions $f_j(v^{(i)})=v^{(i)}_j$ and avoid the dependence on $r$. Applying \Cref{thm:logstar} $k$ times, each time with privacy budget $\eps/k$, failure probability $1/d$ and with the previously selected categories removed from the candidate set, we obtain the following guarantee. 
\begin{theorem}
Let $\eps\in(0,1]$, $n\ge2$ and $k\le d$. There is an $\eps$-DP algorithm in the SVT model that outputs distinct categories $j_1,\ldots,j_k$ such that, with probability at least $1-k/d$, for every $t\in[k]$ the count $v_{j_t}$ is at least the $t$-th largest count minus $O\left(\frac{k\log(d\log n)}{\eps}\right)$. The algorithm makes $O(k(d\log^\ast n+\log^3n))$ SVT queries in expectation.
\end{theorem}
\begin{proof}
In round $t$ the candidate set excludes only the $t-1$ previously selected categories and thus its maximum count is at least the $t$-th largest count overall. By \Cref{thm:logstar} with budget $\eps/k$ and failure probability $1/d$, with probability at least $1-1/d$ the selected category has count within $O(k\log(d\log n)/\eps)$ of this maximum. The claim follows by a union bound over the $k$ rounds and basic composition.
\end{proof}

\paragraph{Sparse mean estimation.}
Let $X_1,\ldots,X_n\in\R^d$ be independent samples from $\calN(\mu,I_d)$, where the mean $\mu$ is $k$-sparse (that is, $\|\mu\|_0\le k$ for some $k\ll d$) and satisfies $\|\mu\|_\infty\le1$. The goal is to estimate $\mu$. Without privacy, the standard estimator computes the empirical mean and zeroes out all but its $k$ entries of largest magnitude, which gives $\ell_2$ error $O(\sqrt{k\log(d)/n})$ with high probability. We assume Gaussian data and a known bound on the entries of $\mu$ for simplicity of presentation. They allow us to clip the entries of the samples to a bounded range that, with high probability, does not change the samples. Some bound on the range of $\mu$ is necessary for this since the SVT API only accepts bounded contributions. Both assumptions can be relaxed using standard techniques (for example, by first computing a coarse private estimate of the range of every coordinate).

Our estimator follows the non-private one. We first privately select $k$ coordinates whose empirical sums have the largest magnitude, and then privately estimate the sums of the selected coordinates. Both steps use only sums of bounded functions of the samples, which we obtain by clipping every entry to the interval $[-B,B]$ for $B=1+\sqrt{2\log(8nd/\beta)}$, where $\beta$ is the failure probability. We denote the clipped samples by $\widetilde X_1,\ldots,\widetilde X_n$.

To select coordinates we use the $2d$ score functions $f_{j,\sigma}(\widetilde X_i)=\frac{B+\sigma\widetilde X_{ij}}{2B}\in[0,1]$ for $j\in[d]$ and $\sigma\in\{-1,+1\}$, so that $\max_{\sigma}\sum_{i}f_{j,\sigma}(\widetilde X_i)=\frac{n}{2}+\frac{1}{2B}\left|\sum_{i}\widetilde X_{ij}\right|$. We select $k$ coordinates in $k$ rounds. In every round we run \Cref{alg:filter_fallback} with privacy budget $\eps/(2k)$ and failure probability $\min\{\beta/(4k),1/(2d)\}$, and then remove both score functions of the selected coordinate from the candidate set. We denote the set of selected coordinates by $\widehat S$. For every $j\in\widehat S$, we then estimate $\sum_{i}\widetilde X_{ij}$ by applying \Cref{cor:gbs:estimate} to the function $g_j(\widetilde X_i)=\frac{B+\widetilde X_{ij}}{2B}$ with privacy budget $\eps/(2k)$ and failure probability $\beta/(4k)$. If $\hat g_j$ is the midpoint of the returned interval, we output $\widehat\mu_j=\frac{B(2\hat g_j-n)}{n}$. For $j\notin\widehat S$ we output $\widehat\mu_j=0$. By basic composition, this algorithm is $\eps$-differentially private.

\begin{theorem}
Let $\eps\in(0,1]$, $\beta\in(0,1)$, $n\ge2$ and $k\le d$. The algorithm above is $\eps$-DP in the SVT model, makes $O(k(d\log^\ast n+\log^3n))$ SVT queries in expectation and, with probability at least $1-\beta$, its output satisfies
\[\|\widehat\mu-\mu\|_\infty=O\left(\sqrt{\frac{\log(d/\beta)}{n}}+\frac{k\sqrt{\log(nd/\beta)}\,\log\frac{d\log n}{\beta}}{n\eps}\right),\qquad\qquad\|\widehat\mu-\mu\|_2\le\sqrt{2k}\,\|\widehat\mu-\mu\|_\infty.\]
For $\delta\in(0,1/2]$, using advanced composition instead of basic composition, the factor $k$ in the second term can be replaced by $O(\sqrt{k\log(1/\delta)})$ at the cost of $(\eps,\delta)$-DP.
\end{theorem}
The analysis follows the non-private one. With high probability no entry is clipped and all the empirical sums are within $A=\sqrt{2n\log(8d/\beta)}$ of their expectations. Therefore every selected coordinate is estimated with error $O(A)$ plus the error of the private estimation. A coordinate in the support of $\mu$ that is not selected must have a small mean, since a coordinate outside of the support, whose empirical sum has magnitude at most $A$, was selected instead.
\begin{proof}
Let $S_j=\sum_{i\in[n]}X_{ij}$. Since $S_j-n\mu_j\sim\calN(0,n)$ and $X_{ij}-\mu_j\sim\calN(0,1)$, standard Gaussian tail bounds and a union bound imply that, with probability at least $1-\beta/2$, $|S_j-n\mu_j|\le A$ for all $j\in[d]$ and $\widetilde X_i=X_i$ for all $i\in[n]$. By \Cref{thm:logstar} and a union bound over the $k$ rounds, with probability at least $1-\beta/4$, the coordinate $j_t$ selected in round $t$ satisfies $|S_{j_t}|\ge|S_j|-E_1$ for every coordinate $j$ that was not selected before round $t$, where $E_1=O\left(\frac{kB}{\eps}\log\frac{d\log n}{\beta}\right)$ is $2B$ times the error bound of \Cref{thm:logstar}. By \Cref{cor:gbs:estimate}, with probability at least $1-\beta/4$ we have $|n\widehat\mu_j-S_j|\le E_2$ for all $j\in\widehat S$, where $E_2=O\left(\frac{kB}{\eps}\log\frac{k\log n}{\beta}\right)$. We assume that all these events hold, which happens with probability at least $1-\beta$.

For $j\in\widehat S$ we have $|\widehat\mu_j-\mu_j|\le(|n\widehat\mu_j-S_j|+|S_j-n\mu_j|)/n\le(E_2+A)/n$. Now let $j^\ast$ be a coordinate in the support of $\mu$ that is not in $\widehat S$. Since $|\widehat S|=k\ge\|\mu\|_0$, some selected coordinate $j'$ is outside of the support of $\mu$, and $j^\ast$ was a candidate in the round in which $j'$ was selected. Therefore
\[A\ge|S_{j'}|\ge|S_{j^\ast}|-E_1\ge n|\mu_{j^\ast}|-A-E_1,\]
that is, $|\widehat\mu_{j^\ast}-\mu_{j^\ast}|=|\mu_{j^\ast}|\le(2A+E_1)/n$. All other coordinates satisfy $\widehat\mu_j=\mu_j=0$. This gives the bound on $\|\widehat\mu-\mu\|_\infty$ since $B=O(\sqrt{\log(nd/\beta)})$. The bound on $\|\widehat\mu-\mu\|_2$ follows since $\widehat\mu-\mu$ is supported on at most $2k$ coordinates. Finally, the bound on the number of queries follows from \Cref{thm:logstar} and \Cref{cor:gbs:estimate}.
\end{proof}

\paragraph{Multiplicative weights and exponential mechanism (MWEM).}
We next consider MWEM, a framework introduced by \cite{HardtLM12} based on multiplicative weights and the exponential mechanism, for privately producing synthetic datasets that accurately answer any set of queries on a protected dataset.
In particular, we assume there exists a set $Q$ of queries $q$, where each query $q: \mathcal{D} \to [-1,+1]$ is a function from the record domain $\mathcal{D}$ to the interval $[-1,+1]$, and extended to datasets $A$ by
\[q(A) = \sum_{x \in \calD} q(x) \cdot A(x).\]
The goal is to produce a synthetic dataset $B$ such that $q(B)\approx q(A)$ for all $q\in Q$. 

We first recall a standard approach toward achieving this goal in the central model of differential privacy, as depicted in \Cref{alg:mwem}. 
The algorithm, known as MWEM, maintains a synthetic dataset that is initialized uniformly over the domain $\calD$ and is iteratively refined.  
At each iteration, it privately selects a query whose current synthetic answer has the largest deviation from the true dataset, measures the true answer with added Laplace noise for privacy, and updates the synthetic dataset multiplicatively to better match this noisy measurement.  
After $T$ iterations, the resulting averaged synthetic dataset $A$ approximates the answers to all queries in $Q$ while preserving differential privacy.

\begin{algorithm}[!htb]
\caption{MWEM Algorithm \cite{HardtLM12}}
\label{alg:mwem}
\begin{algorithmic}[1]
\Require{Private dataset $B$ with $n$ records and $d$ features, query set $Q$, number $T$ of iterations, privacy budget $\eps>0$}
\Ensure{Synthetic histogram $A$ of total mass $n$ over $\calD$}
\State{Initialize $A_0$ as a uniform distribution over the domain $\calD$, scaled by $n$}
\For{$t=1$ to $T$}
\State{Use exponential mechanism to privately select the query $q_t \in Q$ that maximizes the score function $|q(B) - q(A_{t-1})|$, with privacy budget $\frac{\eps}{2T}$}    
\State{$m_t \gets q_t(B) + \Lap\left(\frac{4T}{\eps}\right)$}
\State{Let $A_t$ be $n$ times the distribution whose entries satisfy 
\[A_t(x) \propto A_{t-1}(x) \times \exp\left(q_t(x) \times \frac{m_t - q_t(A_{t-1})}{2n}\right)\]}
\EndFor
\State{\textbf{Return} $A=\frac{1}{T}\sum_{t\in[T]}A_t$}
\Comment{The average of the $T$ updated synthetic histograms}
\end{algorithmic}
\end{algorithm}

Observe that the key challenge to implementing \Cref{alg:mwem} in distributed models is implementing the exponential mechanism to privately select the query $q_t \in Q$ that maximizes the score function $|q(B) - q(A_{t-1})|$. 
Suppose there exists a mapping $f:\mathbb{R}^d\to\mathbb{R}^{d'}$ so that there exists a bijection between queries $q\in Q$ and coordinates $k\in[d']$ so that the query $q$ can be decomposed as the sum $q(B)=\sum_{i\in[n]}(f(v_i))_k$. 
However, this does not quite allow us to sample using score function $|q(B)-q(A_{t-1})|$ due to the term $q(A_{t-1})$. 
Thus, for the embedded dimension $d'$, we use the public offset vector $w\in\mathbb{R}^{d'}$ so that for each query $q$, there exists $k\in[d']$ such that
\[q(B)-q(A_{t-1})=w_k+\sum_{i\in[n]}(f(v_i))_k.\]

For example, consider the class of two-way marginal queries, defined as follows. 
A two-way marginal query $q_{(a,b)}$ is specified by a pair of attributes $a,b\in[d]$, and returns the number of records in the dataset whose attribute $a$ and attribute $b$ are both equal to $1$, i.e., $q_{(a,b)}(B)=\sum_{i\in[n]}\mathbf{1}[v_i(a)=v_i(b)=1]$. 
We define the embedding as follows.  
Let each user record $v_i\in\{0,1\}^d$ encode its attribute values, and let the embedded dimension $d'$ index all pairs $(a,b)$, so that $d'=\O{d^2}$.  
Define $f(v_i)\in\mathbb{R}^{d'}$ by 
\[(f(v_i))_{(a,b)}=\mathbf{1}[v_i(a)=v_i(b)=1].\]
Then for any two-way marginal query $q_{(a,b)}$,
\[q_{(a,b)}(B)=\sum_{i\in[n]} (f(v_i))_{(a,b)}.\]
Moreover, since $A_{t-1}$ is fixed at round $t$, we can introduce a public offset vector $w\in\mathbb{R}^{d'}$ defined by $w_{(a,b)} = -q_{(a,b)}(A_{t-1})$, so that
\[q_{(a,b)}(B)-q_{(a,b)}(A_{t-1}) = w_{(a,b)} + \sum_{i\in[n]} (f(v_i))_{(a,b)}\]
Thus, for two-way marginal queries, the signed MWEM discrepancy can be realized as a coordinate of an aggregated embedded vector plus a public offset. After analyzing \Cref{alg:mwem}, we show how this representation allows us to implement the selection step in the SVT model, using approximate selection in place of the exponential mechanism.

We define $\Delta_t=\max_{q\in Q}|q(A_{t-1})-q(B)|$ and, for a confidence parameter $\beta\in(0,1)$, $\eta=\frac{8T}{\eps}\ln\frac{2T|Q|}{\beta}$.
\begin{restatable}{lemma}{lemmwemmultadd}
\label{lem:mwem:mult:add}
With probability at least $1-\beta$, we have for all $t\in[T]$,
\[|q_t(A_{t-1})-q_t(B)|\ge\Delta_t-\eta,\qquad |m_t-q_t(B)|\le\eta.\]
\end{restatable}
\begin{proof}
Since $q(x)\in[-1,+1]$, replacing one record changes $q(B)$, and hence the score $|q(B)-q(A_{t-1})|$, by at most $2$ (conditioned on the preceding transcript, $A_{t-1}$ is fixed). By the standard utility guarantee of the exponential mechanism \cite{McSherryT07,DworkRBook}, with privacy parameter $\frac{\eps}{2T}$ and score sensitivity $2$, the selected query $q_t$ satisfies $|q_t(A_{t-1})-q_t(B)|\ge\Delta_t-\frac{8T}{\eps}\left(\ln|Q|+\ln\frac{2T}{\beta}\right)=\Delta_t-\eta$ with probability at least $1-\frac{\beta}{2T}$. Further, $m_t-q_t(B)\sim\Lap\left(\frac{4T}{\eps}\right)$ and therefore $|m_t-q_t(B)|\le\frac{4T}{\eps}\ln\frac{2T}{\beta}\le\eta$ with probability at least $1-\frac{\beta}{2T}$. The claim follows by a union bound over the $2T$ events.
\end{proof}

Following the standard analysis of MWEM \cite{HardtLM12}, we show in \Cref{app:mwem} that on the accuracy event of \Cref{lem:mwem:mult:add} the averaged synthetic dataset is accurate.
\begin{restatable}{lemma}{lemmwemutil}
\label{lem:mwem:util}
For an integer $T\ge1$ and $|\calD|\ge2$, on the simultaneous accuracy event of \Cref{lem:mwem:mult:add},
\[\max_{q\in Q}|q(A)-q(B)|\preceq n\cdot\sqrt{\frac{\log|\calD|}{T}}+\eta.\]
\end{restatable}

For $0<\eps\le1$, let $L=\ln\frac{2n|Q|}{\beta}$ and $T=\min\left\{n,\left\lceil (n\eps)^{2/3}(\log|\calD|)^{1/3}L^{-2/3}\right\rceil\right\}$. 
If $\log|\calD|>n$ then the bound below is larger than $n$ and holds trivially. 
Otherwise $T$ is the balancing value of the two terms in \Cref{lem:mwem:util} rounded up to an integer, $T\le n$ and hence $\eta\le\frac{8T}{\eps}L$, and we obtain:
\begin{corollary}
For $0<\eps\le1$, with probability at least $1-\beta$,
\[\max_{q\in Q}|q(A)-q(B)|\preceq n^{2/3}\cdot\left(\frac{\log|\calD|\ln(2n|Q|/\beta)}{\eps}\right)^{1/3}.\]
\end{corollary}

Finally, we show that \Cref{alg:mwem} is differentially private. 
\begin{restatable}{lemma}{lemmwempriv}
\Cref{alg:mwem} is $\eps$-differentially private. 
\end{restatable}
\begin{proof}
The algorithm makes $T$ invocations of the Exponential Mechanism with privacy parameter $\frac{\eps}{2T}$, and $T$ invocations of the Laplace Mechanism with privacy parameter $\frac{\eps}{2T}$.
By the basic composition of differential privacy, the overall algorithm is $\eps$-differentially private.
\end{proof}

We now describe how to implement \Cref{alg:mwem} in the SVT model. The algorithm accesses the private dataset only to select the query $q_t$ and to measure $m_t$. We implement the selection step using our selection algorithms and the measurement step using the estimation of a single sum (\Cref{cor:gbs:estimate}). As in the example of two-way marginals above, the only complication is that the score of a query depends on the synthetic answer $q(A_{t-1})$. This answer is determined by the outputs of the previous rounds and is therefore public in round $t$. Thus it only shifts the thresholds of the SVT queries.

Specifically, for the selection step we use $2|Q|$ candidates, one for every query $q\in Q$ and sign $\sigma\in\{-1,+1\}$, with the functions $f_{q,\sigma}(x)=\frac{1+\sigma q(x)}{2}\in[0,1]$ and the public offsets $w_{q,\sigma}=\frac{n-\sigma q(A_{t-1})}{2}\in[0,n]$. The shifted score of candidate $(q,\sigma)$ is then
\[\sum_{i\in[n]}f_{q,\sigma}(x_i)+w_{q,\sigma}=n+\frac{\sigma\left(q(B)-q(A_{t-1})\right)}{2}\in[0,2n],\]
and the maximum of the shifted scores over all candidates is $n+\Delta_t/2$. Our selection algorithms access the scores only by comparing them with thresholds, and comparing the shifted score of candidate $(q,\sigma)$ with a threshold $\tau$ is the SVT query $(f_{q,\sigma},\tau-w_{q,\sigma})$. Therefore they apply to shifted scores, with the range $[0,n]$ of the scores replaced by $[0,2n]$. We run \Cref{alg:geometric_bs} on all $2|Q|$ candidates with privacy budget $\frac{\eps}{2T}$, failure probability $\frac{\beta}{2T}$ and $\lceil\log_2(2n)\rceil$ rounds. By \Cref{thm:gbs}, with probability at least $1-\frac{\beta}{2T}$ the selected query satisfies $|q_t(A_{t-1})-q_t(B)|\ge\Delta_t-\eta'$, where $\eta'=O\left(\frac{T}{\eps}\log\frac{T|Q|\log n}{\beta}\right)$.

For the measurement step we apply \Cref{cor:gbs:estimate} to the function $\frac{1+q_t(x)}{2}$ with privacy budget $\frac{\eps}{2T}$ and failure probability $\frac{\beta}{2T}$, and set $m_t=2\hat g-n$, where $\hat g$ is the midpoint of the returned interval. Then $|m_t-q_t(B)|\le\eta'$ with probability at least $1-\frac{\beta}{2T}$ (for a suitable constant in the definition of $\eta'$). Therefore \Cref{lem:mwem:mult:add} holds with $\eta'$ in place of $\eta$. The rest of the analysis uses only the two inequalities of \Cref{lem:mwem:mult:add} and thus \Cref{lem:mwem:util} and the corollary above also hold with $\eta'$ in place of $\eta$ (for $T\le n$, $\eta'=O\left(\frac{T}{\eps}\log\frac{n|Q|}{\beta}\right)$). The resulting algorithm is $\eps$-DP by basic composition of its $2T$ private steps. We remark that when $\frac{\beta}{2T}\le\frac{1}{2|Q|}$, the selection step can also be implemented using \Cref{alg:filter_fallback} (with $2n$ in place of $n$), which gives the same bound on $\eta'$ up to a constant factor using $O(|Q|\log^\ast n+\log^3n)$ SVT queries per round in expectation.

Putting these statements together, we have the following:

\begin{theorem}
\label{thm:mwem:pure}
For any private dataset $B$ with $n$ records, finite query set $Q$, privacy parameter $\eps\in(0,1]$ and $\beta\in(0,1)$, there exists an $\eps$-differentially private algorithm in the SVT model that releases a synthetic dataset $A$ such that with probability at least $1-\beta$, we have
\[\max_{q\in Q}|q(A)-q(B)|\preceq n^{2/3}\cdot\left(\frac{\log|\calD|\cdot\log(n|Q|/\beta)}{\eps}\right)^{1/3}.\]
\end{theorem}

In fact, we can achieve better utility bounds if we permit $(\eps,\delta)$-differential privacy for $0<\eps\le1$ and $\delta\in(0,1/2]$, by a similar argument, allowing additional privacy loss per time step and instead using advanced composition: each of the $2T$ private steps can use privacy parameter $\Theta\left(\eps/\sqrt{T\log(1/\delta)}\right)$, which replaces the factor $T$ in $\eta'$ by $O(\sqrt{T\log(1/\delta)})$, and balancing the two terms of \Cref{lem:mwem:util} then gives $T\asymp n\eps\sqrt{\log|\calD|/\log(1/\delta)}/\log(n|Q|/\beta)$:

\begin{theorem}
\label{thm:mwem:approximate}
For any private dataset $B$ with $n$ records, query set $Q$, privacy parameters $0<\eps\le1$ and $\delta\in(0,1/2]$, and $\beta\in(0,1)$, there exists an $(\eps,\delta)$-differentially private algorithm in the SVT model that releases a synthetic dataset $A$ such that with probability at least $1-\beta$, we have
\[\max_{q\in Q}|q(A)-q(B)|\preceq\sqrt{\frac{n}{\eps}}\cdot\polylog\left(|\calD|,|Q|,n,\frac{1}{\delta},\frac{1}\beta\right).\]
\end{theorem}

\section*{Acknowledgments}
AI tools were used to assist in improving the clarity and presentation of some proofs. 
All mathematical arguments were reviewed and verified by the authors.

\bibliography{references}
\bibliographystyle{alpha}

\newpage
\appendix

\section{Proof of the Lower Bound for Selection with Gaussian Aggregation}
\label{sec:lower-bound-proof}

Our lower bound relies on a lower bound for solving the problem in the standard setting of distributed multi-party protocols. Specifically, we use the standard number-in-hand blackboard communication model. Each user initially knows only its own input and its private randomness. The users communicate by writing messages to a common transcript, where each message may depend on the sender's input, the sender's private randomness, public randomness, and the previous transcript. The full transcript is denoted by $M$; we include public randomness and any randomness used for the final decision in $M$, so the output $\widehat Z$ is a function of $M$.
Throughout this appendix, $B(2)$ denotes the uniform distribution on $\{1,2\}$.

\begin{theorem}
\label{sbm-mi-lower-bound}
For every $2/3$-accurate protocol for \probref{sbm-distinguishing}, its transcript $M$ satisfies
\[
    I\left(M; X^{\le n} \,\middle|\, V, Z\right) = \Omega(d).
\]
\end{theorem}

Our proof follows the approach in \cite{braverman2016communication}, which established a lower bound for the sparse Gaussian mean estimation problem (see also \cite{duchirogers2019lower}). The key lemma in their result is Jayram's Hellinger cut-and-paste inequality that we state below, where $h$ denotes Hellinger distance. We use the normalization $h^2(P,Q)=\frac12\int(\sqrt{p}-\sqrt{q})^2$ and natural logarithms in the KL divergence. With this normalization, for $R=(P+Q)/2$ we have $\KL(P\|R)\ge 2h^2(P,R)$ and therefore $\frac12\KL(P\|R)+\frac12\KL(Q\|R)\ge h^2(P,R)+h^2(Q,R)\ge\frac12\left(h(P,R)+h(Q,R)\right)^2\ge\frac12 h^2(P,Q)$, where the last step is the triangle inequality for $h$.

\begin{lemma}[Jayram's cut-and-paste inequality {\cite[Theorem 7]{Jayram09}; see also \cite[Theorem A.1]{braverman2016communication}}]
\lemlab{sbm-jayram-general}
There is an absolute constant $C>0$ such that the following holds. Consider any $n$-user blackboard protocol. For each user $i$, let $\nu^i_1,\nu^i_2$ be two distributions over that user's input space. For $r\in\{1,2\}^n$, let $P_r$ be the transcript distribution when the users' inputs are sampled independently as $X^i\sim \nu^i_{r_i}$. Then, for any $a,b\in\{1,2\}^n$,
\[
    \sum_{i=1}^n h^2\left(P_{a^{i\leftarrow b_i}}, P_a\right) \ge C\cdot h^2\left(P_b, P_a\right).
\]
\end{lemma}

\begin{proof}[Proof of Theorem~\ref{sbm-mi-lower-bound}]
Let $X^{\le n}_j=(X^1_j,\ldots,X^n_j)$. We first relate the information in the randomly chosen coordinate to the information in the full input under $Z=2$.
\begin{equation}
\begin{aligned}
    I\left(M; X^{\le n}_V \,\middle|\, V, Z=2\right) &= \frac{1}{d} \sum_{j=1}^d I\left(M; X^{\le n}_j \,\middle|\, V=j, Z=2\right)\\
        &= \frac{1}{d} \sum_{j=1}^d I\left(M; X^{\le n}_j \,\middle|\, Z=2\right)\\
        &\le \frac{1}{d} I\left(M; X^{\le n} \,\middle|\, Z=2\right).
\end{aligned}
\eqnlab{sbm:coord-to-full}
\end{equation}
The second equality follows from the fact that when $Z=2$, the transcript and inputs are independent of $V$. The inequality follows from the standard chain rule for mutual information \cite[Section~2.5]{coverthomas2006elements}: if $Y_1,\ldots,Y_m$ are mutually independent conditioned on $W$, then $\sum_{\ell=1}^m I\left(Y_\ell; M \,\middle|\, W\right)\le I\left(Y_1,\ldots,Y_m; M \,\middle|\, W\right)$. We apply this with the independent coordinates $X^{\le n}_1,\ldots,X^{\le n}_d$ conditioned on $Z=2$.

Applying the same property to the independent user coordinates $X^1_V,\ldots,X^n_V$ conditioned on $(V,Z=2)$,
\begin{equation}
I\left(M; X^{\le n}_V \,\middle|\, V, Z=2\right) \ge \sum_{i=1}^n I\left(M; X^i_V \,\middle|\, V, Z=2\right).
\eqnlab{sbm:user-decomposition}
\end{equation}

Next fix a user $i\in \{1,\ldots, n\}$ and a coordinate $V=j$, and analyze $I\left(M; X^i_V \,\middle|\, V=j, Z=2\right)$. Let $b$ be a vector in $\{1,2\}^n$. Let $M_{b}$ denote the distribution of the transcript $M$ when, for every user $k$, the input coordinate $X^k_j$ is sampled according to state $b_k$; that is, if $b_k=1$ then $X^k_j=1$, while if $b_k=2$ then $X^k_j$ is sampled uniformly from $\{0,1\}$. All other coordinates are sampled independently and uniformly from $\{0,1\}$. Let $D_b$ denote the corresponding distribution of the full inputs $X^{\le n}$.

First consider the setting where $Z_i\sim B(2)$, user $i$ receives the input $X^i_j$ sampled according to $Z_i$, and the remaining $n-1$ users still receive inputs sampled using $Z=2$. Since all other coordinates are uniform independently of $Z_i$, we have the Markov chain $Z_i \to X^i_j \to M$ in this experiment. By data processing and the comparison between the Jensen-Shannon divergence and the Hellinger distance,
\begin{equation}
\begin{aligned}
I_{Z_i\sim B(2),X^{\le n}\sim D_{\bar 2^{i\leftarrow Z_i}}}\left(M; X^i_j \,\middle|\, V=j\right)&\ge I_{Z_i\sim B(2),X^{\le n}\sim D_{\bar 2^{i\leftarrow Z_i}}}\left(M; Z_i \,\middle|\, V=j\right)\\
&\hspace{-0.5in}=\E_{z\sim B(2)}\left[\KL\left(M_{\bar 2^{i\leftarrow z}} \,\middle\|\, \E_{Z_i} M_{\bar 2^{i\leftarrow Z_i}}\right)\right]\\
&\hspace{-0.5in}=\frac{1}{2}\KL\left(M_{\bar 2^{i\leftarrow 1}} \,\middle\|\, \frac{1}{2}M_{\bar 2^{i\leftarrow 1}}+\frac{1}{2}M_{\bar 2}\right) +\frac{1}{2}\KL\left(M_{\bar 2} \,\middle\|\, \frac{1}{2}M_{\bar 2^{i\leftarrow 1}}+\frac{1}{2}M_{\bar 2}\right)\\
&\hspace{-0.5in}\ge \tfrac12 h^2(M_{\bar 2^{i\leftarrow 1}},M_{\bar 2}) \ .
\end{aligned}
\eqnlab{sbm:one-coordinate-hellinger}
\end{equation}

Observe that, in the experiment with $Z_i\sim B(2)$, $\Pr\left[X^i_j = x \,\middle|\, Z_i=2\right] \ge \frac{2}{3} \Pr\left[X^i_j = x\right]$ for all $x\in\{0,1\}$, so we can apply the following lemma.

\begin{lemma}[{\cite[Lemma 5]{braverman2016communication}}]
\lemlab{sbm:domination}
Suppose two distributions $\mu, \mu'$ satisfy $\mu \ge c\cdot \mu'$. Let $\Pi(X)$ be a random function that only depends on $X$. If $X\sim \mu$ and $X'\sim \mu'$, then
\[
    I(X; \Pi(X)) \ge c\cdot I(X'; \Pi(X')).
\]
\end{lemma}

Applying this lemma gives
\begin{equation}
I_{X^{\le n}\sim D_{\bar 2}}\left(M; X^i_j \,\middle|\, V=j, Z=2\right)\ge \frac{2}{3} I_{Z_i\sim B(2),X^{\le n}\sim D_{\bar 2^{i\leftarrow Z_i}}}\left(M; X^i_j \,\middle|\, V=j\right).
\eqnlab{sbm:lemma-i2}
\end{equation}

We next use \lemref{sbm-jayram-general} to pass from changing one user to changing all users. For the fixed coordinate $j$, apply \lemref{sbm-jayram-general} with $a=\bar 2$ and $b=\bar 1$ to the product input distributions defining $M_b$. This gives
\begin{equation}
    \sum_{i=1}^n h^2 (M_{\bar 2^{i\leftarrow 1}}, M_{\bar 2}) \ge C\cdot h^2 (M_{\bar 1}, M_{\bar 2}).
\eqnlab{sbm:jayram}
\end{equation}
Taking expectation over random $V=j$ and combining the above inequalities, we get
\begin{equation}
\begin{aligned}
    \E_{V} \left[C\cdot h^2 (M_{\bar 1}, M_{\bar 2})\right] &\le \E_{V} \left[\sum_{i=1}^n h^2 (M_{\bar 2^{i\leftarrow 1}}, M_{\bar 2})\right]\\
    &\le 2\E_V\left[\sum_{i=1}^n I_{Z_i\sim B(2),X^{\le n}\sim D_{\bar 2^{i\leftarrow Z_i}}}\left(M; X^i_V \,\middle|\, V\right)\right]\\
    &\le 3\sum_{i=1}^n I\left(M; X^i_V \,\middle|\, V, Z=2\right)\\
    &\le \frac{3}{d} I\left(M; X^{\le n} \,\middle|\, Z=2\right).
\end{aligned}
\eqnlab{sbm:combined}
\end{equation}
The first inequality in \eqnref{sbm:combined} is \eqnref{sbm:jayram}, averaged over the random coordinate $V$. 
The second inequality applies the one-coordinate Jensen-Shannon/Hellinger bound \eqnref{sbm:one-coordinate-hellinger} to each user $i$ and coordinate $V$. 
The third inequality applies \lemref{sbm:domination} as in \eqnref{sbm:lemma-i2}: since $X^i_V$ under $Z=2$ dominates its distribution in the mixed experiment with $Z_i\sim B(2)$ by a factor $2/3$, the mixed mutual information is at most $\frac{3}{2}$ times the mutual information under $Z=2$. 
The last inequality uses \eqnref{sbm:user-decomposition} and \eqnref{sbm:coord-to-full}.

It remains to use the $2/3$-accuracy assumption. 
For a fixed coordinate $j$, let $P_j$ be the distribution of $M$ conditioned on $V=j$ and $Z=1$, and let $Q$ be the distribution of $M$ conditioned on $Z=2$. 
Under $Z=2$, the transcript distribution is independent of $V$, so $Q$ is also the distribution denoted by $M_{\bar 2}$ for every $j$, while $P_j$ is the distribution denoted by $M_{\bar 1}$ for that $j$. 
Let $\bar P=\E_{V}P_V$ be the transcript distribution conditioned on $Z=1$. If $A=\{m:\widehat Z(m)=1\}$, then $2/3$-accuracy gives
\begin{align*}
\bar P(A)-Q(A) \ge \frac{1}{3},
\end{align*}
and hence $\TV(\bar P,Q)\ge 1/3$. Therefore
\begin{align*}
\frac{1}{3}\le \TV(\bar P,Q)\le \E_V\left[\TV(P_V,Q)\right]\le \sqrt{2}\,\E_V\left[h(P_V,Q)\right],
\end{align*}
By Jensen's inequality this implies that
\begin{equation}
\E_V\left[h^2(M_{\bar 1},M_{\bar 2})\right] = \Omega(1).
\eqnlab{sbm:accuracy-hellinger}
\end{equation}

Combining \eqnref{sbm:combined} and \eqnref{sbm:accuracy-hellinger} gives
\[
    I\left(M; X^{\le n} \,\middle|\, Z=2\right) = \Omega(d).
\]
Finally, conditioned on $Z=2$, the transcript $M$ and input $X^{\le n}$ are independent of $V$, so
\[
    I\left(M; X^{\le n} \,\middle|\, V,Z=2\right)=I\left(M; X^{\le n} \,\middle|\, Z=2\right).
\]
Since $Z\sim B(2)$,
\[
    I\left(M; X^{\le n} \,\middle|\, V,Z\right)\ge \frac{1}{2}I\left(M; X^{\le n} \,\middle|\, V,Z=2\right)= \Omega(d),
\]
as claimed.
\end{proof}

To relate the lower bound in Theorem~\ref{sbm-mi-lower-bound} to private aggregation we recall that differential privacy limits mutual information.

For this purpose we will need only privacy with respect to KL divergence (or Renyi DP of order 1). We say that an algorithm $\calA : \calX^n \to Y$ is $\rho$-KL stable if for any pair of neighboring datasets $S,S'$, we have
$\KL(\calA(S)\|\calA(S'))\le \rho$.


\begin{lemma}[\cite{feldmansteinke2017generalization}]
\lemlab{kl-stable-mi}
    Assume that the transcript of a communication protocol viewed as the output of an algorithm acting on the users' inputs is $\rho$-KL stable. Then for any 
    distribution $D$ over $\calX$, if the users' inputs $X^{\le n}$ are sampled i.i.d. from $D$ then we have
    \[ I\left(M; X^{\le n}\right) \le \rho n. \]
\end{lemma}

We can now complete the proof of Theorem~\ref{thm:main-lower-bound} using Lemma \ref{lem:local-zcdp}.

\begin{proof}[Proof of Theorem~\ref{thm:main-lower-bound}]
Let $\calA$ be a $2/3$-accurate algorithm for \probref{sbm-distinguishing} that has access to $\calG(S)$ and satisfies compositional $\rho$-zCDP. By \lemref{local-zcdp}, there is a local $n\rho$-zCDP blackboard protocol $\mbox{Dist-}\calA$ whose output has the same distribution as the output of $\calA$ on every input dataset. In particular, when run on the random input $X^{\le n}$ from \probref{sbm-distinguishing}, $\mbox{Dist-}\calA$ is also $2/3$-accurate. Let $M$ denote its full transcript, including public randomness and the final decision randomness as above.

By Theorem~\ref{sbm-mi-lower-bound},
\[
    I\left(M;X^{\le n}\,\middle|\,V,Z\right)=\Omega(d).
\]
On the other hand, the local $n\rho$-zCDP guarantee implies that the transcript is $n\rho$-KL stable with respect to changing one user's input. Fix any $v\in[d]$ and $z\in\{1,2\}$. Conditioned on $(V,Z)=(v,z)$, the users' inputs are i.i.d. samples from $D_{z,v}$. Hence \lemref{kl-stable-mi} gives
\[
    I\left(M;X^{\le n}\,\middle|\,V=v,Z=z\right)\le n^2\rho .
\]
Averaging over $V$ and $Z$ yields
\[
    I\left(M;X^{\le n}\,\middle|\,V,Z\right)\le n^2\rho .
\]
Combining the two inequalities gives $n^2\rho=\Omega(d)$, and therefore
\[
    n=\Omega\left(\sqrt{\frac d\rho}\right),
\]
as claimed.
\end{proof}

\section{Analysis of MWEM in the Central Model}
\label{app:mwem}
We prove \Cref{lem:mwem:util} using the standard potential argument for MWEM \cite{HardtLM12}, with $\Delta_t$ and $\eta$ as defined before \Cref{lem:mwem:mult:add}.
We define the relative entropy by 
\[\Phi_t=\sum_{x\in\calD}\frac{B(x)}{n}\cdot\log\frac{B(x)}{A_t(x)}.\]
We utilize the following two standard properties of relative entropy:
\begin{fact}[Non-negativity of relative entropy]
\cite{HardtLM12}
\label{fact:entropy:nonneg}
$\Phi_t\ge 0$ for all $t\in[T]$
\end{fact}
\begin{fact}[Jensen's inequality for relative entropy]
\cite{HardtLM12}
\label{fact:entropy:jensen}
$\Phi_0\le\log|\calD|$
\end{fact}
We also have the following well-known fact for the exponential function:
\begin{fact}
\cite{HardtLM12}
\label{fact:exp:bound}
For all $x\in[-1,1]$, we have $\exp(x)\le1+x+x^2$. 
\end{fact}
To bound the error $|q_t(A)-q_t(B)|$ in each round, we show that the relative entropy decreases by some amount in each round.
\begin{lemma}
\label{lem:mwem:phi:diff}
For each $t\in[T]$, we have
\[\Phi_{t-1}-\Phi_t\ge\left(\frac{q_t(A_{t-1})-q_t(B)}{2n}\right)^2-\left(\frac{m_t-q_t(B)}{2n}\right)^2.\]
\end{lemma}
\begin{proof}
By definition, we have
\[\Phi_{t-1}-\Phi_t=\sum_{x\in\calD}\frac{B(x)}{n}\log\frac{A_t(x)}{A_{t-1}(x)}.\]
Let $\alpha_t:=\frac{m_t-q_t(A_{t-1})}{2n}$ and $\beta_t$ be the normalization factor at step $t$. 
Then we have $\frac{A_t(x)}{A_{t-1}(x)}=\frac{1}{\beta_t}\cdot\exp(q_t(x)\cdot\alpha_t)$ so that
\[\Phi_{t-1}-\Phi_t=\frac{\alpha_t\cdot q_t(B)}{n}-\log\beta_t.\]
Since we have the normalization factor 
\[\beta_t=\sum_{x\in\calD}\frac{A_{t-1}(x)}{n}\cdot\exp(q_t(x)\cdot\alpha_t),\]
Note that $A_{t-1}/n$ is a probability distribution over $\calD$ since $\sum_xA_{t-1}(x)=n$. Hoeffding's lemma, applied to $q_t(X)\in[-1,1]$ for $X\sim A_{t-1}/n$, gives, for every real $\alpha_t$,
\[
\log\beta_t\le \frac{\alpha_t q_t(A_{t-1})}{n}+\frac{\alpha_t^2}{2}
\le \frac{\alpha_t q_t(A_{t-1})}{n}+\alpha_t^2.
\]
Therefore,
\[\Phi_{t-1}-\Phi_t\ge\frac{\alpha_t}{n}\cdot\left(q_t(B)-q_t(A_{t-1})\right)-\alpha_t^2.\]
The desired claim then holds from the definition of $\alpha_t$. 
\end{proof}

\lemmwemutil*
\begin{proof}
Firstly, observe that since $\Delta_t=\max_{q\in Q}|q(A_{t-1})-q(B)|$, then 
\begin{align*}
\max_{q\in Q}|q(A)-q(B)|
&\le \frac1T\sum_{t=1}^T\Delta_t
+\frac1T\max_{q\in Q}|q(A_T)-q(A_0)|\\
&\le \frac1T\sum_{t=1}^T\Delta_t+\frac{2n}{T}.
\end{align*}
On the accuracy event of \Cref{lem:mwem:mult:add} we have that $\Delta_t\le|q_t(A_{t-1})-q_t(B)|+\eta$ and $|m_t-q_t(B)|\le\eta$ for all $t\in[T]$.
By \Cref{lem:mwem:phi:diff}, we have
\[\frac{1}{T}\sum_{t\in[T]}\Delta_t\preceq\left(\frac{n^2}{T}\cdot\sum_{t\in[T]}(\Phi_{t-1}-\Phi_t)+\eta^2\right)^{1/2}+\eta.\]
By a telescoping argument, observe that $\sum_{t\in[T]}(\Phi_{t-1}-\Phi_t)=\Phi_0-\Phi_T$.
Then by \Cref{fact:entropy:nonneg} and \Cref{fact:entropy:jensen}, we have
\[\frac{1}{T}\sum_{t\in[T]}\Delta_t\preceq\left(\frac{n^2\cdot\log(|\calD|)}{T}+\eta^2\right)^{1/2}+\eta.\]
Since $\sqrt{x+y}\le\sqrt{x}+\sqrt{y}$ and $2n/T=O(n\sqrt{\log|\calD|/T})$ for $|\calD|\ge2$ and $T\ge1$, putting things together gives
\[\max_{q\in Q}|q(A)-q(B)|\preceq n\cdot\sqrt{\frac{\log|\calD|}{T}}+\eta.\]
\end{proof}

\end{document}

%% file: preamble.tex
\pdfoutput=1
\usepackage[T1]{fontenc}
\usepackage{lmodern}
\usepackage{slantsc}
\usepackage[protrusion=true,expansion=true]{microtype}
\usepackage{breakcites}
\usepackage{amsmath,amssymb,amsfonts,amsthm}
\usepackage{subcaption}
\usepackage{graphicx}
\usepackage{fullpage}
\usepackage{setspace}
\usepackage[backref=page]{hyperref}
\usepackage[nameinlink]{cleveref}
\usepackage{color}
\usepackage{wrapfig}
\usepackage{tikz}
\usetikzlibrary{decorations.pathreplacing}
\usepackage{algorithm}
\usepackage[noend]{algpseudocode}
\usepackage[framemethod=tikz]{mdframed}
\usepackage{xspace}
\usepackage{pgfplots}
\usepackage{framed}
\usepackage{thmtools}
\usepackage{thm-restate}
\usepackage{tabu}
\usepackage{fancyhdr}
\usepackage{booktabs}
\pgfplotsset{compat=1.5}

\newtheorem{theorem}{Theorem}[section]
\newtheorem{corollary}[theorem]{Corollary}
\newtheorem{lemma}[theorem]{Lemma}

\newtheorem{definition}[theorem]{Definition}

\newtheorem{fact}[theorem]{Fact}

\newtheorem{problem}[theorem]{Problem}

\newenvironment{proofof}[1]{\begin{trivlist} \item {\bf Proof
#1:~~}}
  {\qed\end{trivlist}}

\newcommand{\namedref}[2]{\hyperref[#2]{#1~\ref*{#2}}}

\newcommand{\lemlab}[1]{\label{lem:#1}}
\newcommand{\lemref}[1]{\namedref{Lemma}{lem:#1}}

\newcommand{\eqnlab}[1]{\label{eq:#1}}
\newcommand{\eqnref}[1]{\namedref{Equation}{eq:#1}}

\newcommand{\probref}[1]{\namedref{Problem}{prob:#1}}
\newcommand{\problab}[1]{\label{prob:#1}}

\newcommand{\PPr}[1]{\ensuremath{\mathbf{Pr}\left[#1\right]}}

\renewcommand{\O}[1]{\ensuremath{O\left(#1\right)}}

\newcommand{\eps}{\varepsilon}
\newcommand{\R}{\ensuremath{\mathbb{R}}}
\newcommand{\Low}{\texttt{Low}}
\newcommand{\High}{\texttt{High}}
\newcommand{\Close}{\texttt{Close}}

\def \calA    {\mdef{\mathcal{A}}}
\def \calB    {\mdef{\mathcal{B}}}

\def \calD    {\mdef{\mathcal{D}}}
\def \calE    {\mdef{\mathcal{E}}}

\def \calG    {\mdef{\mathcal{G}}}

\def \calN    {\mdef{\mathcal{N}}}

\def \calX    {\mdef{\mathcal{X}}}
\def \calY    {\mdef{\mathcal{Y}}}

\newcommand{\mdef}[1]{{\ensuremath{#1}}\xspace}  

\DeclareMathOperator*{\polylog}{polylog}
\DeclareMathOperator*{\poly}{poly}

\DeclareMathOperator*{\Lap}{Lap}

\newcommand{\E}[2][]{\mdef{\underset{#1}{\mathbb{E}}\left[#2\right]}} 

\newcommand{\ignore}[1]{}

\newif\ifnotes\notestrue 
\ifnotes
\newcommand{\samson}[1]{\textcolor{blue}{{\bf (Samson:} {#1}{\bf ) }} \marginpar{\tiny\bf
             \begin{minipage}[t]{0.5in}
               \raggedright S:
            \end{minipage}}}
\newcommand{\david}[1]{\textcolor{purple}{{\bf (David:} {#1}{\bf ) }} \marginpar{\tiny\bf
             \begin{minipage}[t]{0.5in}
               \raggedright D:
            \end{minipage}}} 
\else
\newcommand{\samson}[1]{}
\newcommand{\david}[1]{}
\fi

\makeatletter
\renewcommand*{\@fnsymbol}[1]{\textcolor{mahogany}{\ensuremath{\ifcase#1\or *\or \dagger\or \ddagger\or
 \mathsection\or \triangledown\or \mathparagraph\or \|\or **\or \dagger\dagger
   \or \ddagger\ddagger \else\@ctrerr\fi}}}
\makeatother

\providecommand{\email}[1]{\href{mailto:#1}{\nolinkurl{#1}\xspace}}

\definecolor{mahogany}{rgb}{0.75, 0.25, 0.0}
\definecolor{darkblue}{rgb}{0.0, 0.0, 0.55}
\definecolor{darkpastelgreen}{rgb}{0.01, 0.75, 0.24}
\definecolor{darkgreen}{rgb}{0.0, 0.2, 0.13}
\definecolor{darkgoldenrod}{rgb}{0.72, 0.53, 0.04}
\definecolor{darkred}{rgb}{0.55, 0.0, 0.0}
\definecolor{forestgreenweb}{rgb}{0.13, 0.55, 0.13}
\definecolor{greencss}{rgb}{0.0, 0.5, 0.0}
\definecolor{bleudefrance}{rgb}{0.19, 0.55, 0.91}
\definecolor{darkspringgreen}{rgb}{0.09, 0.45, 0.27}

\hypersetup{
     colorlinks   = true,
     citecolor    = mahogany,
	 linkcolor	  = darkspringgreen
}

\fancypagestyle{pg}
{
\lhead{}
\rhead{}
\cfoot{--\ \thepage\ --}

}

\AtBeginDocument{%
  \DeclareFontShape{T1}{lmr}{m}{scit}{<->ssub*lmr/m/scsl}{}%
}
\renewcommand{\E}{\mathop{\mathbf E}}

\newcommand{\TV}{\mbox{TV}}
\newcommand{\KL}{\mbox{KL}}

\newcommand{\vx}{\ensuremath{\mathbf{x}}}